\documentclass[review]{siamonline1116}

\usepackage{enumerate}
\usepackage{bbm}
\usepackage{float}
\newcommand{\esssup}{\operatorname{ess\,sup}}
\newcommand{\Var}{\mathrm{Var}}
\newtheorem{assumption}[theorem]{Assumption}
\newtheorem{remark}[theorem]{Remark}
\makeatletter
\newcommand*{\rom}[1]{\expandafter\@slowromancap\romannumeral #1@}
\makeatother

\usepackage{tikz}
\newcommand*\circled[1]{\tikz[baseline=(char.base)]{
            \node[shape=circle,draw,inner sep=2pt] (char) {#1};}}

\makeatletter
\newenvironment{breakablealgorithm}
  {% \begin{breakablealgorithm}
   \begin{center}
     \refstepcounter{algorithm}% New algorithm
     \hrule height.8pt depth0pt \kern2pt% \@fs@pre for \@fs@ruled
     \renewcommand{\caption}[2][\relax]{% Make a new \caption
       {\raggedright\textbf{\ALG@name~\thealgorithm} ##2\par}%
       \ifx\relax##1\relax % #1 is \relax
         \addcontentsline{loa}{algorithm}{\protect\numberline{\thealgorithm}##2}%
       \else % #1 is not \relax
         \addcontentsline{loa}{algorithm}{\protect\numberline{\thealgorithm}##1}%
       \fi
       \kern2pt\hrule\kern2pt
     }
  }{% \end{breakablealgorithm}
     \kern2pt\hrule\relax% \@fs@post for \@fs@ruled
   \end{center}
  }
\makeatother

\usepackage{lipsum}
\usepackage{amsfonts}
\usepackage{graphicx}
\usepackage{epstopdf}
\usepackage{algorithmic}
\ifpdf
  \DeclareGraphicsExtensions{.eps,.pdf,.png,.jpg}
\else
  \DeclareGraphicsExtensions{.eps}
\fi

\numberwithin{theorem}{section}

\newcommand{\TheTitle}{Uncertain Volatility Models with Stochastic Bounds} 
\newcommand{\TheAuthors}{Jean-Pierre Fouque and Ning Ning}

\headers{\TheTitle}{\TheAuthors}

\title{{\TheTitle}\thanks{Submitted  February 15, 2017.
\funding{Research supported by NSF grant DMS-1409434.}}}

\author{
  Jean-Pierre Fouque \thanks{Department of Statistics \& Applied Probability, University of California, Santa Barbara, CA 93106-3110
    (\email{fouque@pstat.ucsb.edu, ning@pstat.ucsb.edu}).}
    \and Ning Ning \footnotemark[2]
}

\usepackage{amsopn}

\ifpdf
\hypersetup{
  pdftitle={\TheTitle},
  pdfauthor={\TheAuthors}
}
\fi

\nolinenumbers
\begin{document}

\maketitle

% REQUIRED
\begin{abstract}
  In this paper, we propose the uncertain volatility models with stochastic bounds. Like the regular uncertain volatility models, we know only that the true model lies in a family of progressively measurable and bounded processes, but instead of using two deterministic bounds, the uncertain volatility fluctuates between two stochastic bounds generated by its inherent stochastic volatility process. This brings better accuracy and is consistent with the observed volatility path such as for the VIX as a proxy for instance. We apply the regular perturbation analysis upon the worst case scenario price, and derive the first order approximation in the regime of slowly varying stochastic bounds. The original problem which involves solving a fully nonlinear PDE in dimension two for the worst case scenario price, is reduced to solving a nonlinear PDE in dimension one and a linear PDE with source, which gives a tremendous computational advantage. Numerical experiments show that this approximation procedure performs very well, even in the regime of moderately slow varying stochastic bounds. 
\end{abstract}

% REQUIRED

% REQUIRED

\section{Introduction}
In the standard Black--Scholes model of option pricing \cite{black1973pricing}, volatility is assumed to be known and constant over time, which seems unrealistic. Extensions of the Black--Scholes model to model ambiguity have been proposed, such as the stochastic volatility approach \cite{heston1993closed, hull1987pricing}, the jump diffusion model  \cite{andersen2000jump, merton1976option}, and the uncertain volatility model \cite{avellaneda1995pricing, lyons1995uncertain}. Among  these extensions, the uncertain volatility model has received intensive attention in Mathematical Finance for risk management purpose.

In the uncertain volatility models (UVMs), volatility is not known precisely and assumed between constant upper and lower bounds $\underline{\sigma}$ and $\overline{\sigma}$. These bounds could be inferred from extreme values of the implied volatilities of the liquid options, or from high-low peaks in historical stock- or option-implied volatilities. 

Under the risk-neutral measure, the price process of the risky asset satisfies the following stochastic differential equation (SDE):
\begin{equation}
\label{processX}
dX_t=rX_tdt+\alpha_t X_tdW_t,
\end{equation}
where r is the constant risk-free rate, ($W_t$) is a Brownian motion and the volatility process $(\alpha_t)$ belongs to a family $\mathcal{A}$  of progressively measurable and $[\underline{\sigma}, \overline{\sigma}]$-valued processes.  

When pricing a European derivative  written on the risky asset with maturity $T$ and nonnegative payoff $h(X_T)$, the seller of the contract is interested in the worst-case scenario. By assuming the worst case, sellers are guaranteed coverage against adverse market behavior if the realized volatility belongs to the candidate set. It is known that the worst-case scenario price at time $t < T$ is given by
\begin{equation}
\label{P}
P(t,X_t):=\exp(-r(T-t))\esssup_{\alpha \in \mathcal{A}}\mathbb{E}_t[h(X_T)],
\end{equation}
where $\mathbb{E}_t[\cdot]$ is the conditional expectation given $\mathcal{F}_t$ with respect to the risk neutral measure. 

Following the arguments in stochastic control theory, $P(t,X_t)$ is the viscosity solution to the following Hamilton-Jacobi-Bellman (HJB) equation, which is the generalized
Black--Scholes--Barenblatt (BSB) nonlinear equation in Financial Mathematics,
\begin{equation}
\begin{split}
\partial_t P+r(x\partial_x P-P)+\sup_{\alpha \in [\underline{\sigma}, \overline{\sigma}]}\bigg[\frac{1}{2}x^2\alpha^2\partial_{xx}^2P \bigg]&=0,\\
P(T,x)&=h(x).
\end{split}
\end{equation}

It is well known that the worst case scenario price is equal to its Black--Scholes price with constant volatility $\overline{\sigma}$ (resp. $\underline{\sigma}$)  for convex (resp. concave) payoff function (see \cite{pham2009continuous} for instance). For general terminal payoff functions, an asymptotic analysis of the worst case scenario option prices as the volatility interval degenerates to a single point is given in \cite{fouque2014approximation}.

In fact, for contingent claims with longer maturities, it is no longer consistent with observed volatility to assume that the bounds are constant (for instance by looking at the VIX over years, which is a popular measure of the implied volatility of S\&P 500 index options). 
Therefore, instead of modeling $\alpha_t$ fluctuating between two deterministic bounds, 
\begin{equation*}
\underline{\sigma} \leq \alpha_t \leq \overline{\sigma}, \;\;\;\; \text{for} \;\;\;\; 0\leq t \leq T ,
\end{equation*}
we consider the case that the uncertain
volatility moves between two stochastic bounds,
\begin{equation*}
\underline{\sigma}_t:=d\sqrt{Z_t} \leq \alpha_t \leq \overline{\sigma}_t:=u\sqrt{Z_t}, \;\;\;\; \text{for} \;\;\;\; 0\leq t \leq T ,
\end{equation*}
where $u$ and $d$ are two constants such that $0<d<1<u$, and $Z_t$ can be any other positive stochastic process. 

In this paper, we consider
the general three-parameter CIR process with evolution
\begin{equation}
\label{processZ}
dZ_t=\delta \kappa(\theta-Z_t)dt+\sqrt{\delta}\sqrt{Z_t}dW_t^Z.
\end{equation}
Here, $\kappa$ and $\theta$ are strictly positive parameters with the Feller condition $\theta \kappa \geq \frac{1}{2}$ satisfied to ensure that $Z_t$ stays positive, $\delta$ is a small positive parameter that characterizes the slow variation of the process $Z_t$, and ($W_t^Z$) is a Brownian motion possibly correlated to ($W_t$) driving the stock price, with $d \langle W,W^Z \rangle_t=\rho dt$ for $|\rho|<1$. One realization of these processes is shown in Figure \ref{varyingbounds} with $\delta=.05$. Denoting $\alpha_t:=q_t \sqrt{Z_t}$, then the uncertainty in the volatility can be
absorbed in the uncertain adapted slope
\begin{equation*}
d\leq q_t \leq u, \;\;\;\; \text{for} \;\;\;\; 0\leq t \leq T.
\end{equation*}
%We let $q_t$ adapt to the new filtration $\mathcal{F}_t$ which supports at least Brownian motions $W_t$ and $W_t^Z$,

\begin{figure}[htbp]
  \centering
  \label{varyingbounds}
  \includegraphics[width=5in,height=3.5in]{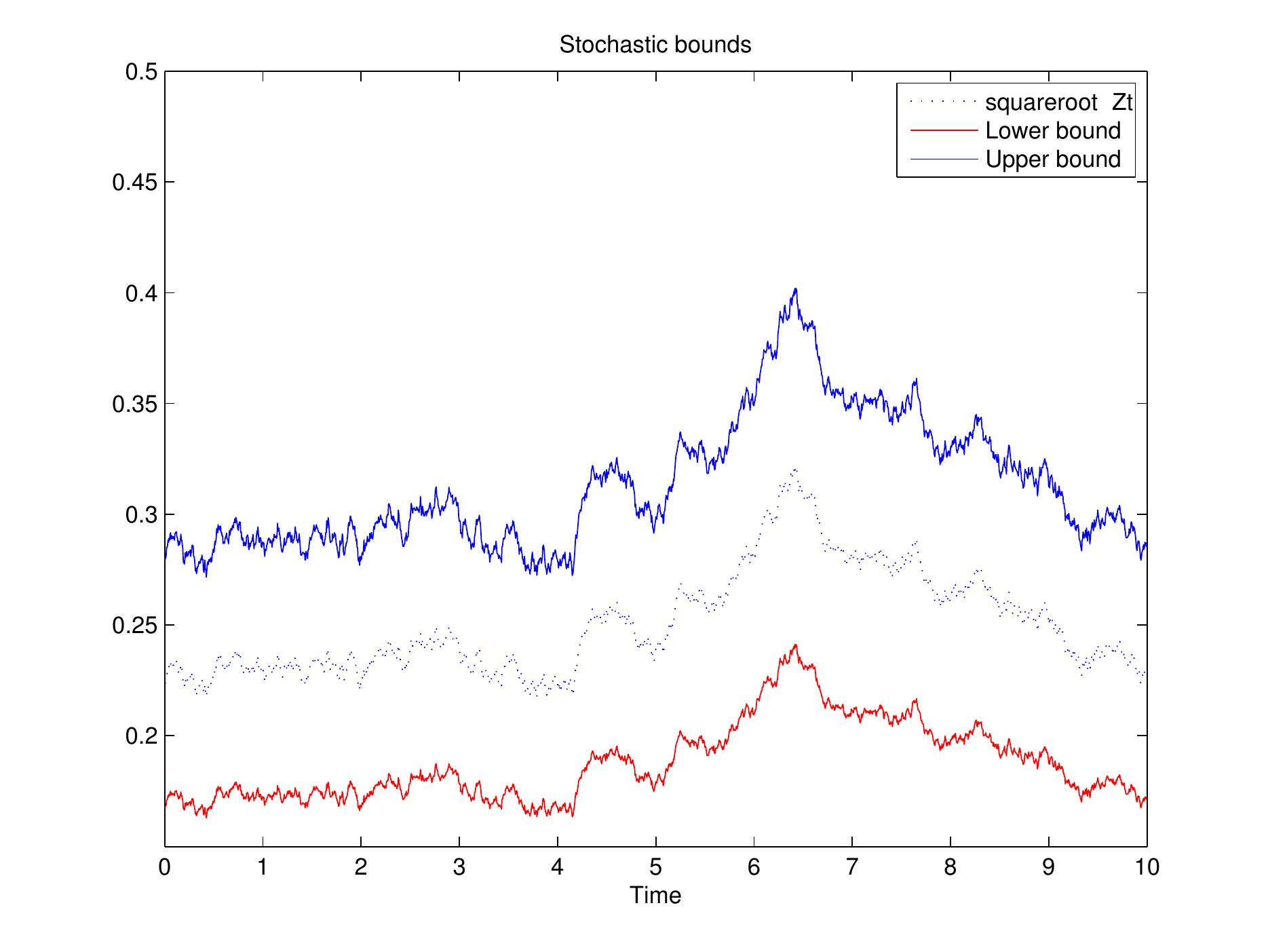}.   
  \caption{\it Simulated stochastic bounds $[0.75\sqrt{Z_t}, 1.25\sqrt{Z_t}]$ where $Z_t$ is the (slow) mean-reverting CIR process \eqref{processZ}.}
\end{figure}

In order to study the asymptotic behavior, we emphasize the importance of $\delta$ and reparameterize the SDE of the risky asset price process as
\begin{equation}
\label{processXdelta}
dX_t^{\delta}=rX_t^{\delta}dt+q_t \sqrt{Z_t} X_t^{\delta}dW_t.
\end{equation}
When $\delta=0$, note that the CIR process $Z_t$ is frozen at $z$, and then the risky asset price process follows the dynamic
\begin{equation}
\label{processX0}
dX_t^{0}=rX_t^{0}dt+q_t \sqrt{z} X_t^{0}dW_t.
\end{equation}	
Both $X_t^{\delta}$ and $X_t^{0}$ start at the same point $x$.

Suppose that $\mathcal{X}$ is a European derivative written on the
risky asset with maturity $T$ and payoff $h(X_T)$. We denote its smallest riskless selling price (worst case scenario) at
time $t<T$ as
\begin{equation}
\label{eq:P_delta}
P^{\delta}(t,x,z):= \exp(-r(T-t))\esssup_{q_{.}\in [d,u]}\mathbb{E}_{(t,x,z)}[h(X_T^{\delta})],
\end{equation}
where
$\mathbb{E}_{(t,x,z)}[\cdot]$ is the conditional expectation given $\mathcal{F}_t$ with $X_t^{\delta}=x$ and $Z_t=z$. When $\delta=0$,  we represent the smallest riskless selling price as 
\begin{equation}
\label{eq:P_0}
P_0(t,x,z)= \exp(-r(T-t))\esssup_{q_{.}\in [d,u]}\mathbb{E}_{(t,x,z)}[h(X_T^{0})],
\end{equation}
where the subscripts in $\mathbb{E}_{(t,x,z)}[\cdot]$ also means that $X_t^{0}=x$ and $Z_t=z$ given the same filtration
$\mathcal{F}_t$. Notice that $P_0(t,X_t,z)$ corresponds to $P(t,X_t)$ in \eqref{P} with constant volatility bounds given by $d\sqrt{z}$ and $u\sqrt{z}$.

The rest of the paper proceeds as follows. In Section \ref{sec:main_result}, we first explain the convergence of the worst
case scenario price $P^{\delta}$ and its second derivative $\partial_{xx}^2 P^{\delta}$ as $\delta$ goes to $0$. We then write down the pricing nonlinear parabolic PDE \eqref{BSB} which characterizes the option price $P^{\delta}(t,x,z)$ as a function of the present time $t$, the value $x$ of the underlying asset, and the levels $z$ of the volatility driving process. At last, we introduce the main result that the first order approximation to $P^{\delta}$ is $P_0+\sqrt{\delta}P_1$ with accuracy in the order of $\mathcal{O}(\delta)$, where we define $P_0$ and $P_1$ in \eqref{P0} and \eqref{P1} respectively. The proof of the main result is presented in Section \ref{sec:proof_main_thm}.  In Section \ref{sec:simulation}, a numerical illustration
is presented. We conclude in Section \ref{sec:conclusion}. Some technical proofs are given in the Appendices.

\section{Main Result}
\label{sec:main_result}
In this section, we first prove the Lipschitz continuity of the worst
case scenario price $P^{\delta}$ with respect to the parameter ${\delta}$. 
Then, we derive the main BSB equation that the worst case scenario price should follow and further identify the first order approximation when ${\delta}$ is small enough. We reduce the original problem of solving the fully nonlinear PDE \eqref{BSB} in dimension two  to solving the nonlinear PDE \eqref{P0} in dimension one  and a linear PDE  \eqref{P1} with source. The accuracy of this approximation is given in Theorem \ref{main_theorem}, the main theorem
of this paper.

\subsection{Convergence of $P^{\delta}$}
\label{sec:convergence}
It is established in Appendix \ref{moment_result} that 
 $X_t$ and $Z_t$ have finite moments uniformly in $\delta$, which leads to the following result:

\begin{proposition} 
\label{X_cvg}
Let $ X^{\delta}$ satisfies the SDE \eqref{processXdelta} and $X^{0}$ satisfies the SDE \eqref{processX0}, then, uniformly in $(q_{\cdot})$,
\begin{equation*}
\mathbb{E}_{(t,x,z)}(X_T^{\delta}- X_T^{0} )^2 \leq C_0\delta 
\end{equation*}
where $C_0$ is a positive constant independent of $\delta$.
\end{proposition}
\begin{proof}
See Appendix \ref{appendix_X_cvg}.
\end{proof}

\begin{assumption}
\label{h_assumptions}
We impose more regularity conditions on the terminal function $h$, i.e.,
Lipschitz continuity, differentiability up to the fourth order and
polynomial
growth conditions on the first four derivatives of $h$:
\begin{equation}
\label{terminal_regularity}
\left\{ \begin{array}{lcl}
|h'(x)| \leq K_1,\\
|h''(x)| \leq K_2(1+|x|^m),\\
|h'''(x) |\leq K_3(1+|x|^n),\\
|h^{(4)}(x)| \leq K_4(1+|x|^l),\\
|h(x)-h(y)|\leq K_0|x-y|.\\
\end{array}\right.
\end{equation}
where $K_i$ for $i \in \{0,1,2,3,4 \}$, $m$, $n$ and $l$ are positive constants. 
\end{assumption}

%\begin{remark}
%The polynomial growth condition on $h^{(4)}$ implies the rest.
%\end{remark}

\begin{theorem} 
\label{P_delta_convergence}
Under Assumption \ref{h_assumptions}, 
$P^{\delta}(t,\cdot,\cdot)$ as a family of functions of $x$ and $z$ indexed by $\delta$, uniformly converge to $P_0(t,\cdot,\cdot)$ in $(q_\cdot)$ with rate $\sqrt{\delta}$,
 for all $(t,x,z)\in 
 [0,T]\times \mathbb{R}^+ \times \mathbb{R}^+$.
\end{theorem}

\begin{proof}
For $P^{\delta}$ given by \eqref{eq:P_delta} and $P_0$ given by \eqref{eq:P_0}, using the Lipschitz continuous of $h(\cdot)$ and by the Cauchy-Schwartz inequality, we have
\begin{equation*}
\begin{split}
|P^{\delta}-P_0| =& \exp(-r(T-t))|\esssup_{q_{.}\in [d,u]}\mathbb{E}_{(t,x,z)}[h(X_T^{\delta})]-\esssup_{q_{.}\in [d,u]}\mathbb{E}_{(t,x,z)}[h(X_T^{0})]|\\
\leq & \exp(-r(T-t))|\esssup_{q_{.}\in [d,u]}\mathbb{E}_{(t,x,z)}[h(X_T^{\delta})]-\mathbb{E}_{(t,x,z)}[h(X_T^{0})]|\\
\leq & \exp(-r(T-t))\esssup_{q_{.}\in [d,u]}|\mathbb{E}_{(t,x,z)}[h(X_T^{\delta})-h(X_T^{0})] |\\
\leq & K_0 \exp(-r(T-t))\esssup_{q_{.}\in [d,u]}\mathbb{E}_{(t,x,z)}|X_T^{\delta}- X_T^{0} |\\
\leq & K_0 \exp(-r(T-t))\esssup_{q_{.}\in [d,u]}[\mathbb{E}_{(t,x,z)}(X_T^{\delta}- X_T^{0} )^2]^{1/2}.
\end{split}
\end{equation*}
Therefore, by Proposition \ref{X_cvg}, we have  
\begin{equation*}
|P^{\delta}-P_0|\leq C_1\sqrt{\delta}
\end{equation*}
where $C_1$ is a positive constant independent of $\delta$, as desired.
\end{proof}

\subsection{Pricing Nonlinear PDEs}
\label{main_result}
We now derive $P_0$ and $P_1$, the leading order term and the first correction for the approximation of the worst-case scenario price $P^{\delta}$,
which
is the solution to the HJB equation associated to the corresponding control problem given by  the generalized BSB nonlinear equation:
\begin{equation}
\label{BSB}
\begin{split}
\partial_t P^{\delta}+r(x\partial_x P^{\delta}-P^{\delta})+\sup_{q\in [d,u]}\{\frac{1}{2}q^2zx^2\partial_{xx}^2 P^{\delta}+\sqrt{\delta}(q \rho z x\partial_{xz}^2 P^{\delta})\}\\
+\delta(\frac{1}{2} z\partial_{zz}^2 P^{\delta}+\kappa(\theta-z)\partial_z
P^{\delta})=0,
\end{split}
\end{equation}
with terminal condition $P^{\delta}(T,x,z)=h(x)$. For simplicity and without loss of generality, $r = 0$ is assumed for
the rest of paper.

In this section, we use the regular perturbation approach to formally expand the value function $P^{\delta}(t,x,z)$ as follows:
\begin{equation}
\label{eq:expansion} 
P^{\delta}=P_0+\sqrt{\delta}P_1+\delta P_2+\cdots.
\end{equation}
Inserting this expansion into the main BSB equation \eqref{BSB}, by Theorem \ref{P_delta_convergence}, the leading order term $P_0$ is the solution to
\begin{equation}
\begin{split}
\label{P0}
\partial_t P_0+\sup_{q\in [d,u]}\{\frac{1}{2}q^2zx^2\partial_{xx}^2
P_0\}&=0,\\
P_0(T,x,z)&=h(x).
\end{split}
\end{equation}
In this case, $z$ is just a positive parameter, we can achieve the existence and uniqueness of a smooth solution to \eqref{P0}, which is referred in the classical work of \cite{friedman1975stochastic} and \cite{pham2009continuous}.

\subsubsection{Convergence of $\partial_{xx}^2P^{\delta}$}
The main references on the regularity for uniformly parabolic equations are \cite{crandall2000lp}, \cite{wang1992regularity1} and  \cite{wang1992regularity2}. In order to use these results, we have to make a log transformation to change variable $x$ to $\ln x$. Then because $q_t$ is bounded away from 0 in $\mathcal{A}$,  \eqref{BSB} is uniformly parabolic. 
Note that given $h$, which satisfies Assumption \ref{h_assumptions}, it is known that $P_0$ belongs to $C_{p}^{1,2}$ ($p$ for polynomial growth).  We conjecture that the result can be extended to $P^{\delta}$ for $\delta$ fixed. Since a full proof is beyond the scope of this paper, here we just assume this property. 

\begin{assumption}
\label{Pdelta_assumptions}
Throughout the paper, we make the following assumptions on $P^{\delta}(t,\cdot,\cdot)$: \\
(i) $P^{\delta}(t,\cdot,\cdot)$ belongs to $C_{p}^{1,2,2 }$ ($p$ for polynomial growth), for $\delta$ fixed.\\
(ii) $\partial_{x}P^{\delta}(t,\cdot,\cdot)$ and $\partial_{xx}^2P^{\delta}(t,\cdot,\cdot)$ are uniformly bounded in $\delta$.
\end{assumption}

Then under this assumption, we have the following Proposition:

\begin{proposition}
\label{second_derivative_cvg}
Under Assumptions \ref{h_assumptions} and \ref{Pdelta_assumptions}, the family $\partial_{xx}^2P^{\delta}(t,\cdot,\cdot)$ of functions of $x$ and $z$ indexed by $\delta$, converges to $\partial_{xx}^2P_0(t,\cdot,\cdot)$  as $\delta$ tends to $0$ with rate $\sqrt{\delta}$, uniformly  on compact sets in $x$ and $z$, for $t \in [0,T)$.
\end{proposition}

\begin{proof}
Under Assumptions  \ref{h_assumptions} and \ref{Pdelta_assumptions}, and by Theorem \ref{P_delta_convergence}, the Proposition can be obtained by following the arguments in Theorem 5.2.5 of \cite{giga2010nonlinear}.
\end{proof}

Denote the zero sets 
%of $\partial_{xx}^2P^{\delta}$ and 
of $ \partial_{xx}^2P_0$ as   
\begin{align*}
%S_{t,z}^{\delta}:=\{x=x(t,z) \in \mathbb{R}^+| \partial_{xx}^2P^{\delta}(t,x,z)=0\} ,\;\;\; 
S_{t,z}^{0}:=\{x=x(t,z) \in \mathbb{R}^+| \partial_{xx}^2P_0(t,x,z)=0\}.
\end{align*}
Define the set where $\partial_{xx}^2P^{\delta}$ and $ \partial_{xx}^2P_0$ take different signs as
\begin{equation}
\label{A_delta_set}
\begin{split}
A_{t,z}^{\delta}:=&\{x=x(t,z)|\partial_{xx}^2P^{\delta}(t,x,z)> 0, \partial_{xx}^2P_0(t,x,z) < 0 \}\\
& \cup \{x=x(t,z)|\partial_{xx}^2P^{\delta}(t,x,z)< 0, \partial_{xx}^2P_0(t,x,z) > 0 \}.
\end{split}
\end{equation}  

\begin{assumption}
\label{zero_points}
We make the following assumptions:\\
(i) There is a finite number of zero points of $\partial_{xx}^2P_0$, for any $t \in [0,T]$. That is, $S_{t,z}^{0}=\{x_1< x_2 < \cdots < x_{m(t)}\}$ and $\max\limits_{0\leq t\leq T}m(t)\leq n$.\\
(ii) There exists a constant $C$ such that the set $A_{t,z}^{\delta}$ defined in \eqref{A_delta_set} is included in $ \cup_{i=1}^n I_i^{\delta}$, where 
  \begin{equation*}
  I_i^{\delta}:=[x_i-C\sqrt{\delta}, x_i+C\sqrt{\delta}], \;\;\;\; \text{for } x_i \in S_{t,z}^{0}\text{ and } 1\leq i\leq m(t).
  \end{equation*}

\end{assumption}

\begin{remark}
\label{remark:thrid_derivative}
Here we explain the rationale for Assumption \ref{zero_points} (ii). 

Suppose $P_0$ has a third derivative with respect to $x$, which does not vanish on the set $S_{t,z}^{0}$. By Proposition \ref{second_derivative_cvg},  $\partial_{xx}^2P^{\delta}(t,\cdot,\cdot)$ converges to $\partial_{xx}^2P_0(t,\cdot,\cdot)$ with rate $\sqrt{\delta}$, therefore
we conclude that there exists a constant $C$ such that on $( \cup_{i=1}^n I_i^{\delta})^c$, $\partial_{xx}^2P^{\delta}(t,\cdot,\cdot)$ and $\partial_{xx}^2P_0(t,\cdot,\cdot)$ have the same sign, and Assumption \ref{zero_points} (ii) would follow. This is illustrated  in Figure \ref{fig:second_derivative} with an example with two zero points for $\partial_{xx}^2P_0(t,\cdot,\cdot)$.

Otherwise, $I_i^{\delta}$ would have a larger radius of  order $\mathcal{O}(\delta^{\alpha})$ for $\alpha \in (0, \frac{1}{2})$, and then the accuracy in the Main Theorem \ref{main_theorem} would be $\mathcal{O}(\delta^{\alpha+1/2})$, but in any case of order 
$o(\sqrt{\delta})$.
%we know that the distance between two zero points of $\partial_{xx}^2P^{\delta}$ and $ \partial_{xx}^2P_0$ should be proportional to the distance of $\partial_{xx}^2P^{\delta}$ and $ \partial_{xx}^2P_0$, hence Assumption \ref{zero_points} (ii) is a direct result; otherwise, $I_i^{\delta}$ would have radius of order $\mathcal{O}(\delta^{\alpha})$ for $\alpha \in (0, \frac{1}{2})$, and then the accuracy in the Main Theorem \ref{main_theorem} would be $\mathcal{O}(\delta^{\alpha+1/2})$.
\end{remark}

\subsubsection{Optimizers}
The optimal control in the nonlinear PDE \eqref{P0} for $P_0$,  denoted as
\begin{equation*}
q^{\ast,0}(t,x,z):=\arg\max_{q\in [d,u]}\{\frac{1}{2}q^2zx^2\partial_{xx}^2
P_0\},
\end{equation*}
is given by
\begin{equation}
\label{q_0} 
 q^{\ast,0}(t,x,z)=\left\{
\begin{array}{rcl} u, \partial_{xx}^2P_0\geq 0 \\
d, \partial_{xx}^2P_0<0\end{array}\right..
\end{equation}

The optimizer to the main BSB equation \eqref{BSB} is given in the following lemma:
\begin{lemma}
	\label{q_ast_lemma}
Under Assumption \ref{zero_points}, for $\delta$ sufficiently small and
	for $x \notin S_{t,z}^{0}$, the optimal control in the nonlinear PDE \eqref{BSB} for $P^\delta$,  denoted as 
	\begin{equation*}
	q^{\ast,\delta}(t,x,z):=\arg\max_{q\in [d,u]}\{\frac{1}{2}q^2zx^2\partial_{xx}^2 P^{\delta}+\sqrt{\delta}(q \rho z x\partial_{xz}^2 P^{\delta})\},
	\end{equation*}
	is given by
	\begin{equation}
	\label{q_ast}
	q^{\ast,\delta}(t,x,z)=\left\{
	\begin{array}{rcl} u, \partial_{xx}^2P^{\delta}\geq 0 \\
	d, \partial_{xx}^2P^{\delta}<0\end{array}\right..
	\end{equation}
\end{lemma}
\begin{proof}
	To find the optimizer $q^{\ast,\delta}$ to
	$$\sup_{q\in [d,u]}\{\frac{1}{2}q^2zx^2\partial_{xx}^2 P^{\delta}+\sqrt{\delta}(q \rho z x\partial_{xz}^2 P^{\delta})\},$$
	we firstly relax the restriction $q\in [d,u]$ to $q\in \mathbb{R}$.
	
	Denote
	$$f(q):=\frac{1}{2}q^2zx^2\partial_{xx}^2 P^{\delta}+\sqrt{\delta}(q \rho z x\partial_{xz}^2 P^{\delta}).$$
By the result of Proposition \ref{second_derivative_cvg} that $\partial_{xx}^2P^{\delta}$ uniformly converge to $\partial_{xx}^2P_0$ as $\delta$ goes to $0$, 
for $x \notin S_{t,z}^{0}$, the optimizer of $f(q)$ is given by 
	$$\hat{q}^{\ast,\delta} =-\frac{\rho \sqrt{\delta} \partial_{xz}^2P^{\delta}}{x\partial_{xx}^2P^{\delta}}.$$
	
Since $X_t$ and $Z_t$ are strictly positive, the sign of the coefficient of $q^2$ in $f(q)$ is determined by the sign of $\partial_{xx}^2 P^{\delta}$. We have the following cases represented in Figure \ref{fig1}, from which we can see that for $\delta$ sufficiently small such that $|\hat{q}^{\ast,\delta}| \leq d$, the optimizer is given by
$$q^{\ast,\delta}=u\mathbbm{1}_{\{\partial_{xx}^2P^{\delta}\geq 0\}} +d\mathbbm{1}_{\{\partial_{xx}^2P^{\delta}< 0\}}. $$

\begin{figure}[htbp]
  \centering
  \label{fig1}
  \includegraphics[width=5in,height=3.5in]{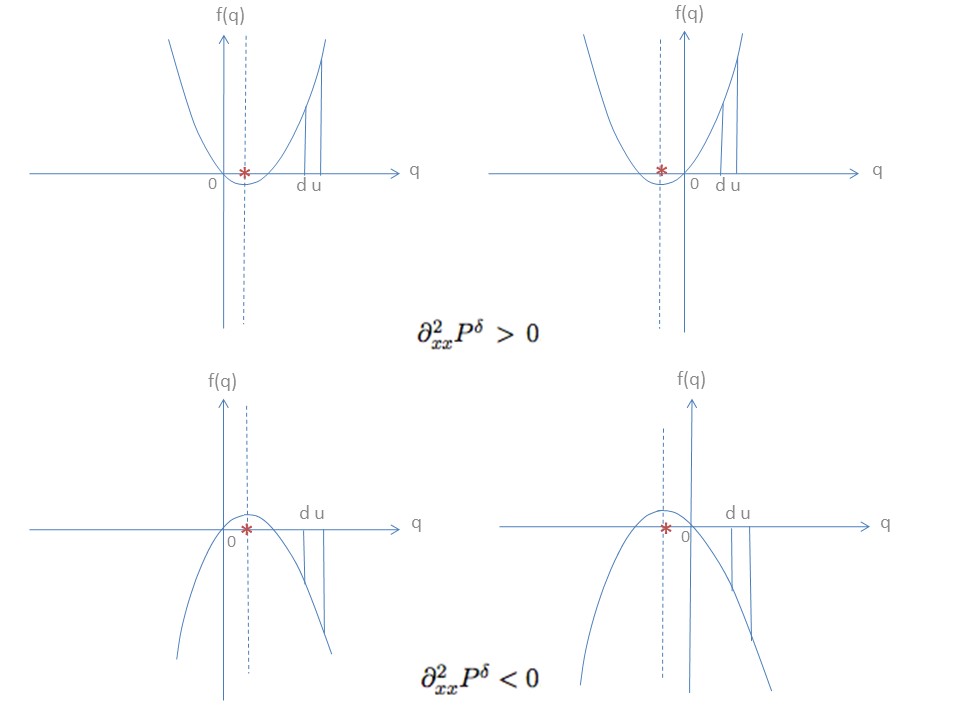}.   
  \caption{\it Illustration of the derivation of $q^{\ast,\delta}$: if $\partial_{xx}^2P^{\delta} \geq 0$, whether $\hat{q}^{\ast,\delta}$ is positive or negative, with the requirement ${q\in [d,u]}$, $q^{\ast,\delta} =u$; otherwise $q^{\ast,\delta} =d$.}
\end{figure}

\end{proof}

\begin{remark}
When $h(\cdot)$ is convex (resp. concave), since supremum and expectation preserves convexity (resp. concavity), one can see that the worst case scenario price 
\begin{equation*}
P^{\delta}(t,x,z)=\exp(-r(T-t))\esssup_{q_{.}\in [d,u]}\mathbb{E}_{(t,x,z)}[h(X_T)],
\end{equation*}
 is convex (resp. concave) with $\partial_{xx}^2P^{\delta}> 0$ (resp. $< 0$), and thus $q^{\ast,\delta}=u$ (resp. $=d$). In these two cases, we are back to perturbations around Black--Scholes prices which have been treated in \cite{fouque2011multiscale}.
In this paper, we work with general terminal payoff functions, neither convex nor concave. 
\end{remark}

Plugging the optimizer $q^{\ast,\delta}$ given by Lemma \ref{q_ast_lemma}, the BSB equation \eqref{BSB} can be rewritten as 
\begin{equation}
\label{BSB_ast}
\begin{split}
\partial_t P^{\delta}+\frac{1}{2}(q^{\ast,\delta})^2zx^2\partial_{xx}^2 P^{\delta}+\sqrt{\delta}( q^{\ast,\delta} \rho z x\partial_{xz}^2 P^{\delta})+\delta(\frac{1}{2} z\partial_{zz}^2 P^{\delta}+\kappa(\theta-z)\partial_z
P^{\delta})=0,
\end{split}
\end{equation}
with terminal condition $P^{\delta}(T,x,z)=h(x)$.

\subsubsection{Heuristic Expansion and Accuracy of the Approximation}
We insert the expansion \eqref{eq:expansion} into the main BSB equation \eqref{BSB_ast} and collect terms in
successive powers of $\sqrt{\delta}$. Under Assumption \ref{zero_points} that $q^{\ast,\delta}\rightarrow q^{\ast,0}$ as $\delta \rightarrow 0$, without loss of accuracy, the first order correction term $P_1$ is chosen as the solution to the linear equation:
\begin{equation}
\begin{split}
\label{P1}
\partial_t P_1+\frac{1}{2}(q^{\ast, 0})^2zx^2\partial_{xx}^2 P_1+q^{\ast, 0}\rho  z x\partial_{xz}^2
P_0&=0,\\
P_1(T,x,z)&=0,
\end{split}
\end{equation}
where $q^{\ast,0}$ is given by \eqref{q_0}.

Since \eqref{P1} is linear, the existence and uniqueness result of a smooth solution $P_1$ can be achieved by firstly change the variable $x \rightarrow \ln x$, and then use the classical result  of  \cite{friedman1975stochastic} for the parabolic equation \eqref{P1} with diffusion coefficient bounded below by $d^2z>0$.
%, and eventually apply the generalized form of Feynman--Kac formula to derive the existence of a specific solution.

Note that that the source term is proportional to the parameter $\rho$.
%, $P_1$ depends on the special parameter $\rho$, which characterizes the relation between the process $X$ and the process $Z$, and it can be seen that $P_1$ changes proportionally to different values of $\rho$. 

We shall show in the following that under additional regularity conditions imposed on the derivatives of $P_0$ and $P_1$, the approximation error $|P^{\delta}-P_0-\sqrt{\delta}P_1|$ is of order $\mathcal{O}(\delta)$.

\begin{assumption}
\label{derivative_assumptions}
We assume polynomial growth for the following derivatives of $P_0$ and $P_1$:
\begin{equation}
\left\{ \begin{array}{ll}
|\partial_{xz}^2P_0(s,x,z)|\leq a_{11} (1+x^{b_{11}}+z^{c_{11}})\\
|\partial_{z}P_0(s,x,z)|\leq a_{01} (1+x^{b_{01}}+z^{c_{01}})\\
|\partial_{xx}^2P_1(s,x,z)|\leq \bar{a}_{20} (1+x^{\bar{b}_{20}}+z^{\bar{c}_{20}})\\
|\partial_{z}P_1(s,x,z)|\leq \bar{a}_{01} (1+x^{\bar{b}_{01}}+z^{\bar{c}_{01}})\\
|\partial_{zz}^2P_1(s,x,z)|\leq \bar{a}_{02} (1+x^{\bar{b}_{02}}+z^{\bar{c}_{02}})
\end{array}\right.
\end{equation}
where $a_i, b_i, c_i, \bar{a}_i, \bar{b}_i, \bar{c}_i$ are positive integers for $i \in (20, 11, 01, 02)$.
\end{assumption}

%\begin{remark}
%Smoothness and differentiability of $P_0$ and $P_1$ can be implied by the regularity of the terminal condition $h(\cdot)$.
%\end{remark}

\begin{theorem}[Main Theorem] 
\label{main_theorem}
Under Assumptions \ref{h_assumptions}, \ref{zero_points} and \ref{derivative_assumptions}, the residual function $E^{\delta}(t,x,z)$ defined by
\begin{equation}
\label{error}
E^{\delta}(t,x,z):= P^{\delta}(t,x,z)-P_0(t,x,z)-\sqrt{\delta}P_1(t,x,z)
\end{equation}
is of order $\mathcal{O}(\delta)$. In other words, $\forall(t,x,z)\in
[0,T]\times \mathbb{R}^+ \times \mathbb{R}^+$, there exists a positive constant
C, such that $|E^{\delta}(t,x,z)|\leq C\delta$, where $C$ may depend on
$(t,x,z)$ but not on $\delta$.
\end{theorem}

Recall that a function $f^{\delta}(t,x,z)$ is of order $\mathcal{O}(\delta^k)$, denoted $f^{\delta}(t,x,z)\sim \mathcal{O}(\delta^k)$, for $\forall(t,x,z)\in
[0,T]\times \mathbb{R}^+ \times \mathbb{R}^+$, there exists a positive constant $C$ depend on
$(t,x,z)$ but not on $\delta$, such that 
$$|f^{\delta}(t,x,z)|\leq C\delta^k.$$
Similarly, we denote $f^{\delta}(t,x,z)\sim o(\delta^k)$, if
$$\limsup_{\delta \rightarrow 0} | f^{\delta}(t,x,z) |/\delta^k=0.$$

\section{Proof of the Main Theorem \ref{main_theorem}}
\label{sec:proof_main_thm}
Define the following operator
\begin{equation}
\begin{split}
\mathcal{L}^{\delta}(q):&=\partial_t +\frac{1}{2}q^2zx^2\partial_{xx}^2 +\sqrt{\delta} q\rho  z x\partial_{xz}^2 +\delta(\frac{1}{2} z\partial_{zz}^2 +\kappa(\theta-z)\partial_z)\\
&=\mathcal{L}_0(q)+\sqrt{\delta}\mathcal{L}_1(q)+\delta\mathcal{L}_2,
\end{split}
\end{equation} 
where the operators $\mathcal{L}_0(q)$, $\mathcal{L}_1(q)$, and $\mathcal{L}_2$ are defined by:
$$\mathcal{L}_0(q):=\partial_t +\frac{1}{2}q^2zx^2\partial_{xx}^2,$$
$$\mathcal{L}_1(q):=q \rho  z x\partial_{xz}^2,$$
$$\mathcal{L}_2:=\frac{1}{2} z\partial_{zz}^2 +\kappa(\theta-z)\partial_z.$$
Note that:

\noindent $\bullet$ $\mathcal{L}_0(q)$ contains the time derivative and is the Black--Scholes operator  $\mathcal{L}_{BS}(q\sqrt{z})$.

\noindent $\bullet$ $\mathcal{L}_1(q)$ contains the mixed derivative due to the covariation between $X$ and $Z$.

\noindent $\bullet$ $\delta\mathcal{L}_2$ is the infinitesimal generator of the process $Z$, also denoted by $\delta \mathcal{L}_{CIR}$.

The main equation \eqref{BSB_ast} can be rewritten as 
\begin{equation}
\label{L_equation_delta}
\begin{split}
\mathcal{L}^{\delta}(q^{\ast,\delta})P^{\delta}&=0,\\
P^{\delta}(t,x,z)&=h(x).
\end{split}
\end{equation}
Equation \eqref{P0} becomes
\begin{equation}
\label{L_equation_0}
\begin{split}
\mathcal{L}_{BS}(q^{\ast,0}) P_0&=0,\\
P_0(T,x,z)&=h(x).
\end{split}
\end{equation}
Equation \eqref{P1} becomes
\begin{equation}
\label{L_equation_1}
\begin{split}
\mathcal{L}_{BS}(q^{\ast,0}) P_1+\mathcal{L}_1(q^{\ast,0})P_0&=0,\\
P_1(T,x,z)&=0.
\end{split}
\end{equation}

%$$\underbrace{u'-P(x)u^2-Q(x)u-R(x)}_{\text{=0, since~$u$ is a particular solution.}}$$

Applying the operator $\mathcal{L}^{\delta}(q^{\ast,\delta})$ to the error term, it follows that
\begin{align*}
\mathcal{L}^{\delta}(q^{\ast,\delta})E^{\delta} =&\mathcal{L}^{\delta}(q^{\ast,\delta})(P^{\delta}-P_0-\sqrt{\delta}P_1)\\
=&\underbrace{\mathcal{L}^{\delta}(q^{\ast,\delta})P^{\delta}}_{\text{=0, eq. \eqref{L_equation_delta}.}}-\mathcal{L}^{\delta}(q^{\ast,\delta})(P_0+\sqrt{\delta}P_1)\\
=&-\left (\mathcal{L}_{BS}(q^{\ast,\delta})+\sqrt{\delta}\mathcal{L}_1(q^{\ast,\delta})+\delta\mathcal{L}_{CIR} \right)(P_0+\sqrt{\delta}P_1)\\
=&-\mathcal{L}_{BS}(q^{\ast,\delta})P_0-\sqrt{\delta} \left [\mathcal{L}_1(q^{\ast,\delta})P_0+\mathcal{L}_{BS}(q^{\ast,\delta})P_1 \right ]-{\delta}\left [\mathcal{L}_1(q^{\ast,\delta})P_1+\mathcal{L}_{CIR}P_0 \right ]\\
&-{\delta}^{\frac{3}{2}} \left[\mathcal{L}_{CIR}P_1 \right ]\\
=&-\underbrace{\mathcal{L}_{BS}(q^{\ast,0}) P_0}_{\text{=0, eq. \eqref{L_equation_0}.}}-(\mathcal{L}_{BS}(q^{\ast,\delta})-\mathcal{L}_{BS}(q^{\ast,0}))P_0-\sqrt{\delta} \Bigg [\underbrace{\mathcal{L}_1(q^{\ast,0})P_0+\mathcal{L}_{BS}(q^{\ast,0}) P_1}_{\text{=0, eq. \eqref{L_equation_1}.}}\\
&+(\mathcal{L}_1(q^{\ast,\delta})-\mathcal{L}_1(q^{\ast,0}))P_0+( \mathcal{L}_{BS}(q^{\ast,\delta}) -\mathcal{L}_{BS}(q^{\ast,0}) )P_1 \Bigg ]\\
&-{\delta} \left[\mathcal{L}_1(q^{\ast,\delta})P_1+\mathcal{L}_{CIR}P_0 \right]-{\delta}^{\frac{3}{2}}(\mathcal{L}_{CIR}P_1 )\\
=& -(\mathcal{L}_{BS}(q^{\ast,\delta})-\mathcal{L}_{BS}(q^{\ast,0}))P_0
-{\sqrt{\delta}} \Bigg [ (\mathcal{L}_1(q^{\ast,\delta})-\mathcal{L}_1(q^{\ast,0}))P_0\\
&+( \mathcal{L}_{BS}(q^{\ast,\delta}) -\mathcal{L}_{BS}(q^{\ast,0}) )P_1 \Bigg ]-{\delta}\left [\mathcal{L}_1(q^{\ast,\delta})P_1+\mathcal{L}_{CIR}P_0 \right ]-{{\delta}^{\frac{3}{2}}}(\mathcal{L}_{CIR}P_1)\\
=&-\frac{1}{2}[(q^{\ast,\delta})^2 -(q^{\ast,0})^2 ]zx^2\partial_{xx}^2 P_0\\
&-{\sqrt{\delta}} \left[\rho  (q^{\ast,\delta}-q^{\ast,0}) z x\partial_{xz}^2P_0+\frac{1}{2}\left((q^{\ast,\delta})^2 -(q^{\ast,0})^2 \right ) zx^2\partial_{xx}^2 P_1 \right]\\
&-{\delta} \left[\rho (q^{\ast,\delta}) z x\partial_{xz}^2P_1+\frac{1}{2} z\partial_{zz}^2P_0 +\kappa(\theta-z)\partial_z P_0 \right]\\
&-{{\delta}^{\frac{3}{2}}} \left[\frac{1}{2} z\partial_{zz}^2P_1 +\kappa(\theta-z)\partial_zP_1 \right],
\end{align*}
where $q^{\ast,0}$ and $q^{\ast,\delta}$ are given in  \eqref{q_0} and  \eqref{q_ast} respectively.

The terminal condition of $E^{\delta}$ is given by
$$E^{\delta}(T,x,z)=P^{\delta}(T,x,z)-P_0(T,x,z)-\sqrt{\delta}P_1(T,x,z)=0.$$ 

\subsection{Feynman--Kac representation of the error term}
For $\delta$ sufficiently small, the optimal choice $q^{\ast,\delta}$ to the main BSB equation  \eqref{BSB} is given explicitly in Lemma \ref{q_ast_lemma}. Correspondingly, the asset price in the worst case scenario is a stochastic process which satisfies the SDE \eqref{processX} with $(q_t)=(q^{\ast,\delta})$ and $r=0$, i.e.,
\begin{equation}
\label{processX_ast}
dX_t^{\ast,\delta}=q^{\ast,\delta}\sqrt{Z_t} X_t^{\ast,\delta}dW_t.
\end{equation} 
Given the existence and uniqueness result of $X_t^{\ast, \delta}$ proved in Appendix  \ref{Existence_and_uniqueness}, we have the following probabilistic representation of $E^{\delta}(t,x,z)$ by Feynman--Kac formula:
\begin{equation*}
E^{\delta}(t,x,z)=\rom{1}_0+\delta^{\frac{1}{2}}\rom{1}_1+\delta\rom{1}_2+\delta^{\frac{3}{2}}\rom{1}_3,
\end{equation*}
where 

\begin{align*}
\rom{1}_0:=\mathbb{E}_{(t,x,z)} \Bigg [&\int_t^T \frac{1}{2}\left ((q^{\ast,\delta})^2 -(q^{\ast,0})^2 \right )Z_s(X_s^{\ast, \delta})^2\partial_{xx}^2 P_0(s, X_s^{\ast, \delta},Z_s)  ds \Bigg ],& \\
\rom{1}_1:=\mathbb{E}_{(t,x,z)} \Bigg [&\int_t^T \Bigg( (q^{\ast,\delta}-q^{\ast,0})\rho Z_s X_s^{\ast, \delta}\partial_{xz}^2P_0(s, X_s^{\ast, \delta},Z_s)\\
&+\frac{1}{2}\left((q^{\ast,\delta})^2 -(q^{\ast,0})^2 \right ) Z_s(X_s^{\ast, \delta})^2\partial_{xx}^2 P_1(s, X_s^{\ast, \delta},Z_s) \Bigg) ds
 \Bigg ],\\
\rom{1}_2:=\mathbb{E}_{(t,x,z)} \Bigg [&\int_t^T \Bigg( q^{\ast, \delta} \rho Z_s X_s^{\ast, \delta}\partial_{xz}^2P_1(s, X_s^{\ast, \delta},Z_s)
+\frac{1}{2} Z_s\partial_{zz}^2P_0(s, X_s^{\ast, \delta},Z_s) \\
&+\kappa(\theta-Z_s)\partial_z P_0(s, X_s^{\ast, \delta},Z_s) \Bigg)ds \Bigg ],\\
\rom{1}_3:=\mathbb{E}_{(t,x,z)}\Bigg [&\int_t^T \Bigg( \frac{1}{2} Z_s\partial_{zz}^2P_1(s, X_s^{\ast, \delta},Z_s) +\kappa(\theta-Z_s)\partial_zP_1(s, X_s^{\ast, \delta},Z_s)\Bigg) ds \Bigg ].
\end{align*}

Note that for $q^{\ast,0}$ given in \eqref{q_0} and $q^{\ast,\delta}$ given in \eqref{q_ast} , we have
\begin{equation}
\label{difference}
q^{\ast,\delta} -q^{\ast,0}= (u-d)( \mathbbm{1}_{\{\partial_{xx}^2P^{\delta}\geq 0\}}-\mathbbm{1}_{\{\partial_{xx}^2P_0\geq 0\}}  ), 
\end{equation}
and
\begin{equation}
\label{square_difference}
(q^{\ast,\delta})^2 -(q^{\ast,0})^2= (u^2-d^2)( \mathbbm{1}_{\{\partial_{xx}^2P^{\delta}\geq 0\}}-\mathbbm{1}_{\{\partial_{xx}^2P_0\geq 0\}}  ) .
\end{equation}  
Also note that $\{q^{\ast,\delta} \neq q^{\ast,0}\}=A_{t,z}^{\delta}$  defined in \eqref{A_delta_set}.

%Also note that for the set $A_{t,z}^{\delta}$ defined in \eqref{A_delta_set}, 
%we see that $q^{\ast,\delta} \neq q^{\ast,0}$ on $A_{t,z}^{\delta}$ and $q^{\ast,\delta} = q^{\ast,0}$ on $(A_{t,z}^{\delta})^c$. 

In order to show that $E^{\delta}$ is of order $\mathcal{O}(\delta)$, it suffices to show that $\rom{1}_0$ is of order $\mathcal{O}(\delta)$, $\rom{1}_1$ is of order $\mathcal{O}(\sqrt{\delta})$, and $\rom{1}_2$ and $\rom{1}_3$ are uniformly bounded in $\delta$. Clearly, $\rom{1}_0$ is the main term that directly determines the order of the error term $E^{\delta}$. 
%However, proving $\rom{1}_0$ of order $\mathcal{O}(\delta)$ has the biggest challenge, since $A_{t,z}^{\delta}$ shrink as $\delta$ goes to $0$, the process $X_t^{\ast, \delta}$ also changing at the same time.  

\subsection{Control of the term $\rom{1}_0$}
In this section, we are going to handle the dependence in $\delta$ of the process $X^{\ast, \delta}$ by a time-change argument.
\begin{theorem}
\label{Control_I0}
Under Assumptions \ref{h_assumptions}, \ref{Pdelta_assumptions} and \ref{zero_points}, there exists a positive constant $M_0$, such that 
\begin{equation*}
|\rom{1}_0|\leq M_0\delta
\end{equation*}
where $M_0$ may depend on
$(t,x,z)$ but not on $\delta$.
That is, $\rom{1}_0$ is of order $\mathcal{O}(\delta)$.
\end{theorem}
\begin{proof}
By Proposition \ref{second_derivative_cvg} and $A_{s,z}^{\delta}$ being compact, there exists a constant $C_0$ such that
\begin{equation*}
|\partial_{xx}^2 P_0(s, X_s^{\ast, \delta},Z_s)|\leq C_0\sqrt{\delta}, \text{ for } X_s^{\ast, \delta} \in A_{s,z}^{\delta}.
\end{equation*}
%where $C_0$ is a positive constant and does not depend on $\delta$. 
Then, since $0<d\leq q^{\ast,\delta}, q^{\ast,0}\leq u$, we have
\begin{align}
\label{eq:I0_indicator_requirement}
|\rom{1}_0|& \leq \mathbb{E}_{(t,x,z)} \Bigg [\int_t^T \frac{1}{2} |(q^{\ast,\delta})^2 -(q^{\ast,0})^2 |Z_s(X_s^{\ast, \delta})^2
|\partial_{xx}^2 P_0(s, X_s^{\ast, \delta},Z_s)|  ds \Bigg ]\nonumber\\
&\leq   \frac{u^2}{2 d^2}C_0\sqrt{\delta}\,\mathbb{E}_{(t,x,z)}\left[\int_t^T \mathbbm{1}_{\{X_s^{\ast, \delta} \in A_{s,z}^{\delta}\}} (q^{\ast,\delta})^2Z_s(X_s^{\ast, \delta})^2 ds \right].
\end{align}
In order to show that $\rom{1}_0$ is of order $\mathcal{O}(\delta)$, it suffices to show that there exists a  constant $C_1$ such that 
\begin{equation}
\label{indicator_difference_eqn_pre}
\mathbb{E}_{(t,x,z)}\left[\int_t^T \mathbbm{1}_{\{X_s^{\ast, \delta} \in A_{s,z}^{\delta}\}}\sigma^2_s ds \right] \leq C_1\sqrt{\delta},
\end{equation}
where 
$\sigma_s:=q^{\ast,\delta} \sqrt{Z_s} X_s^{\ast, \delta}$ and $dX_s^{\ast, \delta}=\sigma_sdW_s$ by \eqref{processX_ast}.
Define the stopping time 
\begin{equation*}
\tau(v):=\inf\{s>t; \langle X^{\ast,\delta}\rangle_s >v \},
\end{equation*}
where 
\begin{equation*}
\langle X^{\ast,\delta}\rangle_s=\int_t^s \sigma^2(X_u^{\ast, \delta}) du.
\end{equation*}
Then according to Theorem 4.6 (time-change for martingales) of \cite{karatzas2012brownian}, we know that
$X^{\ast,\delta}_{\tau(v)}=B_v$
is a standard one-dimensional Brownian motion on $(\Omega, \mathcal{F}_v^B, \mathbb{Q}^B)$.

From the definition of $\tau(v)$ given above, we have
\begin{equation*}
\label{tau_integral}
\int_t^{\tau(v)} \sigma^2(X_s^{\ast, \delta})ds =v,
\end{equation*}
which tells us that the inverse function of $\tau(v)$ is
\begin{equation}
\label{eq:tau_inverse}
\tau^{-1}(T)=\int_t^T \sigma^2(X_s^{\ast, \delta})ds.
\end{equation}
Next use the substitution $s=\tau(v)$ 
and for any $i \in [1,m(v)]$, we have

\begin{equation}
\label{eq1}
\begin{split}
\int_t^T \mathbbm{1}_{\{|X_s^{\ast, \delta}-x_i|<C\sqrt{\delta}\}}\sigma^2(X_s^{\ast, \delta}) ds
=&\int_t^{\tau^{-1}(T)} \mathbbm{1}_{\{|X_{\tau(v)}^{\ast, \delta}-x_i|<C\sqrt{\delta}\}}\sigma^2(X_{\tau(v)}^{\ast, \delta}) d\tau(v)\\
=& \int_t^{\tau^{-1}(T)} \mathbbm{1}_{\{|X_{\tau(v)}^{\ast, \delta}-x_i|<C\sqrt{\delta}\}}\sigma^2(X_{\tau(v)}^{\ast, \delta}) \frac{1}{\sigma^2(X_{\tau(v)}^{\ast, \delta})}dv\\
=& \int_t^{\tau^{-1}(T)} \mathbbm{1}_{\{|X_{\tau(v)}^{\ast, \delta}-x_i|<C\sqrt{\delta}\}}dv\\
=& \int_t^{\tau^{-1}(T)} \mathbbm{1}_{\{|B_v-x_i|<C\sqrt{\delta}\}}dv. 
\end{split}
\end{equation}
Note that on the set $\{|B_v-x_i|<C\sqrt{\delta}\}$, we have
$(X_s^{\ast, \delta})^2\leq (x_i+C\sqrt{\delta})^2\leq D,$
where $D$ is a positive constant,
and then by \eqref{eq:tau_inverse} we have
\begin{equation}
\label{tau_inverse_upper_bound}
\begin{split}
\tau^{-1}(T)&=\int_t^T (q^{\ast,\delta} \sqrt{Z_s} X_s^{\ast, \delta})^2 ds \leq D u^2T \sup_{t\leq s \leq T}Z_s.
\end{split}
\end{equation}
Then from \eqref{eq1} and \eqref{tau_inverse_upper_bound}, by decomposing in $\{\sup_{t\leq s \leq T}Z_s \leq M \}$ and $\{\sup_{t\leq s \leq T}Z_s > M \}$ for any $M>z$, we  obtain 
\begin{equation}
\label{term_sum}
\begin{split}
\mathbb{E}_{(t,x,z)} \bigg [\int_t^{\tau^{-1}(T)} \mathbbm{1}_{\{|B_v-x_i|<C\sqrt{\delta}\}}dv \bigg ] \sim \mathcal{O}(\sqrt{\delta}).
\end{split}
\end{equation}
with details for this last step given in Appendix \ref{Appendix:proof_of_eq_term_sum}.

By finite union over the $x_i$'s we deduce \eqref{indicator_difference_eqn_pre} and  the theorem follows.
\end{proof}

\begin{remark}
As we noted in Remark \ref{remark:thrid_derivative}, if the third derivative of $P_0$ with respect to $x$ vanishes on the set $S_{t,z}^{0}$, the size of $I_i^{\delta}$ is of order $\mathcal{O}(\delta^{\alpha})$ for $\alpha \in (0, \frac{1}{2})$. In that case, \eqref{eq:I0_indicator_requirement} still holds but \eqref{term_sum} would be of order $\mathcal{O}(\delta^{\alpha})$, 
and then the result of Theorem \ref{Control_I0} and the accuracy in the Main Theorem \ref{main_theorem} would be of order $\mathcal{O}(\delta^{\alpha+1/2})$.
\end{remark}

\subsection{Control of the term $\rom{1}_1$}
\begin{theorem}
\label{Control_I1}
Under Assumptions \ref{h_assumptions}, \ref{Pdelta_assumptions}, \ref{zero_points} and \ref{derivative_assumptions}, there exists a constant $M_1$, such that 
\begin{equation*}
|\rom{1}_1|\leq M_1 \sqrt{\delta}
\end{equation*}
where $M_1$ may depend on
$(t,x,z)$ but not on $\delta$.
That is, $\rom{1}_1$ is of order $\mathcal{O}(\sqrt{\delta})$.
\end{theorem}
\begin{proof}
Under Assumption \ref{derivative_assumptions} and $0<d\leq q^{\ast,\delta}, q^{\ast,0}\leq u$, we have
\begin{equation*}
\begin{split}
|\rom{1}_1| &\leq \mathbb{E}_{(t,x,z)} \Bigg [\int_t^T \Bigg( |q^{\ast,\delta}-q^{\ast,0}| \rho Z_s X_s^{\ast, \delta}|\partial_{xz}^2P_0(s, X_s^{\ast, \delta},Z_s)|\\
&\hskip 2cm +\frac{1}{2}|(q^{\ast,\delta})^2 -(q^{\ast,0})^2 |Z_s(X_s^{\ast, \delta})^2|\partial_{xx}^2 P_1(s, X_s^{\ast, \delta},Z_s)|\Bigg) ds \Bigg ]\\
&\leq   \frac{\rho u}{d^2} \mathbb{E}_{(t,x,z)} \Bigg [\int_t^T  \mathbbm{1}_{\{X_s^{\ast, \delta} \in A_{s,z}^{\delta}\}} (q^{\ast,\delta})^2 Z_s X_s^{\ast, \delta}a_{11} (1+(X_s^{\ast, \delta})^{b_{11}}+Z_s^{c_{11}})ds \Bigg ]\\
&\hskip .3cm+\frac{u^2}{2d^2} \mathbb{E}_{(t,x,z)} \Bigg [\int_t^T \mathbbm{1}_{\{X_s^{\ast, \delta} \in A_{s,z}^{\delta}\}}(q^{\ast,\delta})^2 Z_s(X_s^{\ast, \delta})^2\bar{a}_{20} (1+(X_s^{\ast, \delta})^{\bar{b}_{20}}+Z_s^{\bar{c}_{20}}) ds \Bigg ].
\end{split}
\end{equation*}
Using the same techniques in proving Theorem \ref{Control_I0}, the result that $X_s^{\ast, \delta}$ and $Z_s$ have finite moments uniformly in $\delta$,  and $X_s^{\ast, \delta}\leq C(X_s^{\ast, \delta})^2$ on $\{X_s^{\ast, \delta} \in A_{s,z}^{\delta}\}$, we can deduce that $\rom{1}_1$ is of order $\mathcal{O}(\sqrt{\delta})$.
\end{proof}

With the result of theorem \ref{Control_I0} that $\rom{1}_0$ is of order $\mathcal{O}(\delta)$, the result of theorem \ref{Control_I1} that $\rom{1}_1$ is of order $\mathcal{O}(\sqrt{\delta})$, and the result that $\rom{1}_2$ and $\rom{1}_3$ are uniformly bounded in $\delta$ where derivation of these bounds are given in the appendix \ref{sec:Proof_of_Uniform_Boundedness_I2_I3}, we can see that
\begin{equation*}
\begin{split}
E^{\delta}(t,x,z)=\rom{1}_0+\delta^{\frac{1}{2}}\rom{1}_1+\delta\rom{1}_2+\delta^{\frac{3}{2}}\rom{1}_3,
\end{split}
\end{equation*} is of order $\mathcal{O}(\delta)$, which completes the proof of the main Theorem \ref{main_theorem}.

\section{Numerical Illustration}
\label{sec:simulation}
In this section, we use the nontrivial example in \cite{fouque2014approximation}, and
consider a symmetric European butterfly spread with the payoff function
\begin{equation}
\label{butterfly}
h(x)=(x-90)^{+}-2(x-100)^++(x-110)^+
\end{equation}
presented in Figure \ref{fig:SymmetricButterflySpread}. Although this payoff function does not satisfy the conditions imposed in this paper, we could consider a regularization of it, that is to introduce a small parameter for the regularization and then remove this small parameter asymptotically without changing the  accuracy estimate. This  can be achieved by considering $P_0(T-\epsilon, x)$ as the regularized payoff (see \cite{papanicolaou2003singular} for details on this regularization procedure in the context of the Black--Scholes equation).

\begin{figure}[htbp]
  \centering
  \label{fig:SymmetricButterflySpread}
  \includegraphics[width=5in,height=3.5in]{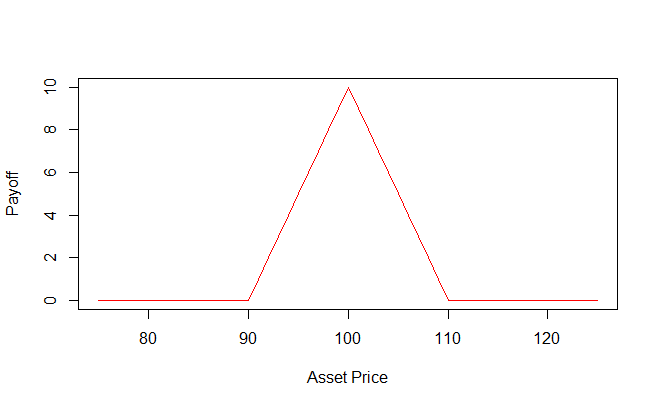}.   
  \caption{{\it The payoff function of a symmetric European butterfly spread.}}
\end{figure}

The original problem is to solve the fully nonlinear PDE \eqref{BSB} in dimension two for the worst case scenario price, which is not analytically solvable in practice. In the following, we use the Crank--Nicolson version of the weighted finite difference method in \cite{guyon2013nonlinear}, which corresponds to the case of solving $P_0$ in one dimension. To extend the original algorithm to our two dimensional case, we apply discretization grids on time and two state variables. Denote $u_{i,j}^n:=P_0(t_n,x_i,z_j)$, $v_{i,j}^n:=P_1(t_n,x_i,z_j)$ and $w_{i,j}^n:=P^{\delta}(t_n,x_i,z_j)$, where $n=0, 1, \cdots, N$ stands for the index of time, $i=0,1,\cdots,I$ stands for the index of the asset price process, and $j=0,1,\cdots,J$ stands for the index of the volatility process. In the following, we build a uniform grid of size $100 \times 100$ and use $20$ time steps.

We use the classical discrete approximations to the continuous derivatives:
\begin{align*}
&\partial_{x}(w_{i,j}^{n})=\frac{w_{i+1,j}^{n}-w_{i-1,j}^{n}}{2\Delta x} &\partial_{zz}^2(w_{i,j}^{n})=\frac{w_{i,j+1}^{n}+w_{i,j-1}^{n}-2w_{i,j}^{n}}{\Delta z^2}\\
&\partial_{xx}^2(w_{i,j}^{n})=\frac{w_{i+1,j}^{n}+w_{i-1,j}^{n}-2w_{i,j}^{n}}{\Delta x^2} & \partial_{z}(w_{i,j}^{n})=\frac{w_{i,j+1}^{n}-w_{i,j-1}^{n}}{2\Delta z}\\
&\partial_{xz}^2(w_{i,j}^{n})=\frac{w_{i+1,j+1}^{n}+w_{i-1,j-1}^{n}-w_{i-1,j+1}^{n}-w_{i+1,j-1}^{n}}{4\Delta x \Delta z} & \partial_{t}(w_{i,j}^{n})=\frac{u_{i,j}^{n+1}-u_{i,j}^n}{\Delta t}
\end{align*}

To simplify our algorithms and facilitate the implementation by matrix operations, we denote the following operators without any parameters:
\begin{align*}
&L_{xx}=zx^2\partial_{xx}^2 & L_{zz}=z\partial_{zz}^2 && L_{xz}=xz\partial_{xz}^2\nonumber\\
&L_{x}=x\partial_{x} & L_{z1}=\partial_{z} && L_{z2}=z\partial_{z}
\end{align*}

\subsection{Simulation of $P_0$ and $P_1$}
Note that in the PDE \eqref{P1} for $P_1$ , $q^{\ast,0}$ must be solved in the PDE \eqref{P0} for $P_0$ . Therefore, we solve $P_0$ and $P_1$ together in each $100\times 100$ space grids and iteratively back to the starting time.\\

\begin{breakablealgorithm}
\caption{Algorithm to solve $P_0$ and $P_1$}
\label{P0_P1_algorithm}
\begin{algorithmic}[1]
\STATE{Set 
$u_{i,j}^N=h(x_I)$ and $v_{i,j}^N=0$.}
\STATE{Solve $u_{i,j}^n$ (predictor)
\begin{equation*}
\begin{split}
\frac{u_{i,j}^{n+1}-u_{i,j}^n}{\Delta t}+[\frac{1}{2}(q_{i,j}^{n+1})^2 L_{xx}] (\frac{u_{i,j}^{n+1}+u_{i,j}^n}{2})=0
\end{split}
\end{equation*}
with 
\begin{equation*}
\begin{split}
q_{i,j}^{n+1}=&u\mathbbm{1}_{\{u^2L_{xx}(u_{i,j}^{n+1})\geq d^2L_{xx}(u_{i,j}^{n+1})\}}\
+ d\mathbbm{1}_{\{u^2L_{xx}(u_{i,j}^{n+1})< d^2L_{xx}(u_{i,j}^{n+1})\}}\
\end{split}
\end{equation*}}

\STATE{Solve $u_{i,j}^n$ (corrector)
\begin{equation*}
\begin{split}
\frac{u_{i,j}^{n+1}-u_{i,j}^n}{\Delta t}+[\frac{1}{2}(q_{i,j}^{n})^2 L_{xx}] (\frac{u_{i,j}^{n+1}+u_{i,j}^n}{2})=0
\end{split}
\end{equation*}
with 
\begin{equation*}
\begin{split}
q_{i,j}^{n}=&u\mathbbm{1}_{\{u^2L_{xx}(\frac{u_{i,j}^{n+1}+u_{i,j}^{n}}{2})\geq d^2L_{xx}(\frac{u_{i,j}^{n+1}+u_{i,j}^{n}}{2})\}}\
+ d\mathbbm{1}_{\{u^2L_{xx}(\frac{u_{i,j}^{n+1}+u_{i,j}^{n}}{2})< d^2L_{xx}(\frac{u_{i,j}^{n+1}+u_{i,j}^{n}}{2})\}}\
\end{split}
\end{equation*}}

\STATE{Solve $v_{i,j}^n$
$$\frac{v_{i,j}^{n+1}-v_{i,j}^n}{\Delta t}+\frac{1}{2}(q_{i,j}^{n})^2 L_{xx}(\frac{v_{i,j}^{n+1}+v_{i,j}^n}{2})+\rho (q_{i,j}^{n})L_{xz}(\frac{u_{i,j}^{n+1}+u_{i,j}^n}{2}) =0$$}
\end{algorithmic}
\end{breakablealgorithm}
\bigskip
Throughout all the experiments, we set $X_0=100$, $Z_0=0.04$, $T=0.25$, $r=0$, $d=0.75$, and $u=1.25$. Therefore, the two deterministic bounds for $P_0$ are given by 
$\underline{\sigma}=d\sqrt{Z_0}=0.15$ and $\overline{\sigma}=u\sqrt{Z_0}=0.25$, which are standard Uncertain Volatility model bounds setup. From Figure \ref{fig:P0_BS_compare}, we can see that $P_0$ is above the Black--Scholes prices with constant volatility $0.15$ and $0.25$ all the time, which corresponds to the fact that we need extra cash to superreplicate the option when facing the model ambiguity. As expected, the Black--Scholes prices with constant volatility  $0.25$ (resp. $0.15$) is a good approximation when $P_0$ is convex (resp. concave).

\begin{figure}[htbp]
  \centering
  \label{fig:P0_BS_compare}
  \includegraphics[width=5in,height=3.5in]{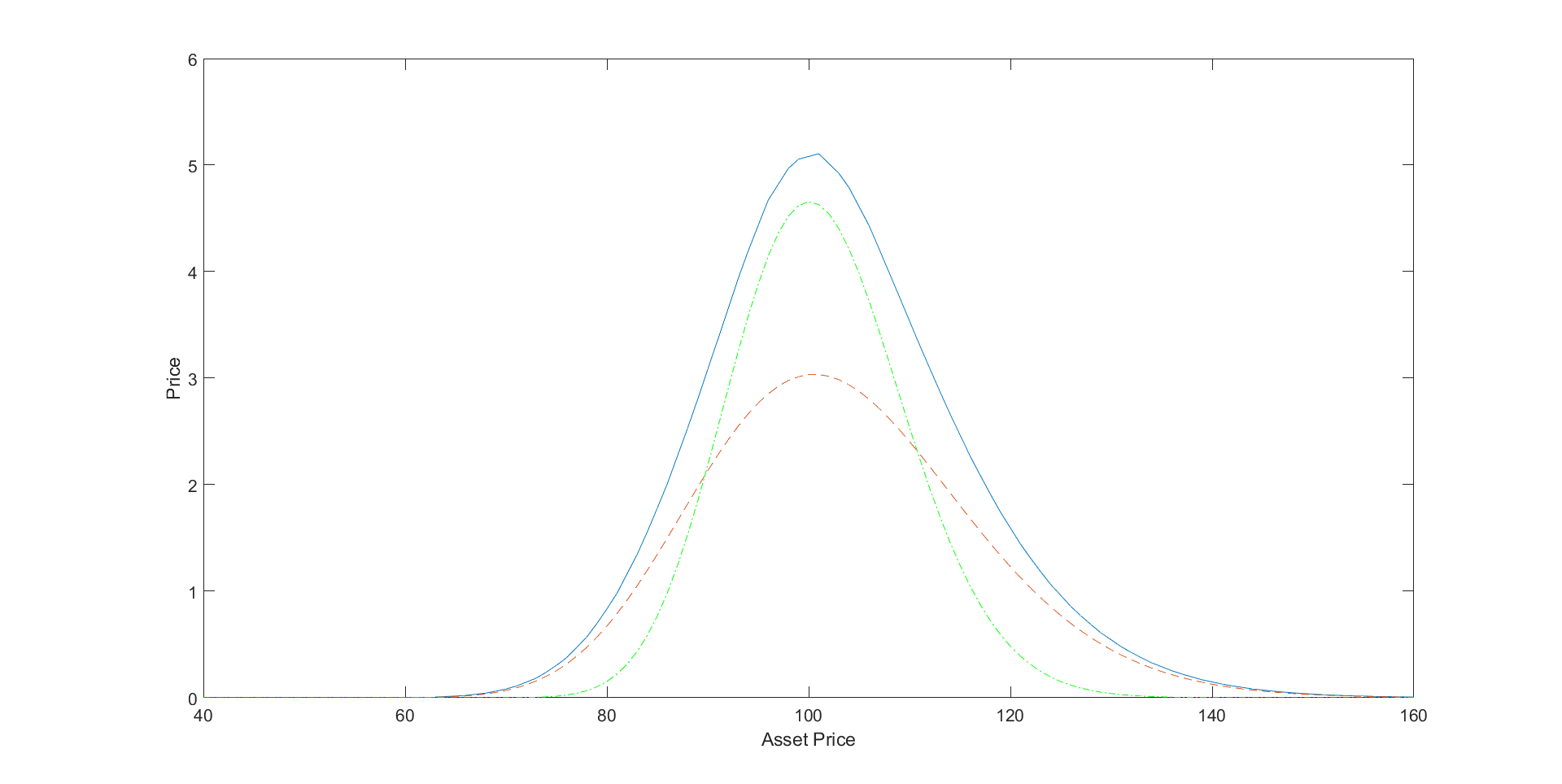}.   
  \caption{\it The blue curve represents the usual uncertain volatility model price $P_0$ with two deterministic bounds $0.15$ and $0.25$, the red curve marked with "- -" represents the Black--Scholes prices with $\sigma=0.25$, the green curve marked with '-.' represents the Black--Scholes prices with $\sigma=0.15$.}
\end{figure}

\subsection{Simulation of $P_\delta$}
Considering the main BSB equation given by \eqref{BSB},
if we relax the restriction 
$q\in [d,u]$ to $q\in \mathbb{R}$, the optimizer of 
$$f(q):=\frac{1}{2}q^2zx^2\partial_{xx}^2 P^{\delta}+ q \rho z x\sqrt{\delta}\partial_{xz}^2 P^{\delta}$$ is given by 
$\hat{q}^{\ast,\delta} =-\frac{\rho \sqrt{\delta} \partial_{xz}^2P^{\delta}}{x\partial_{xx}^2P^{\delta}},$
and the maximum value of $f(q)$ is given by
$f(\hat{q}^{\ast,\delta}) =-\frac{\rho^2 \delta z(\partial_{xz}^2P^{\delta})^2}{2\partial_{xx}^2P^{\delta}}.$
Therefore, 
$$\sup_{q\in [d,u]}f(q)=f(u)\vee f(d)\vee f(\hat{q}^{\ast,\delta}).$$ 

To simplify the algorithm, we denote 
\begin{equation*}
L_A=\frac{1}{2}u^2 L_{xx}+u\rho \sqrt{\delta}L_{xz},
\qquad
L_B=\frac{1}{2}d^2 L_{xx}+d\rho \sqrt{\delta}L_{xz},
\qquad
L_C=-\frac{\rho^2 \delta (L_{xz})^2}{2L_{xx}}.
\end{equation*}

\begin{breakablealgorithm}
\caption{Algorithm to solve $P^{\delta}$}\label{P_delta_algorithm}
\begin{algorithmic}[1]
\STATE{Set 
$w_{i,j}^N=h(x_I)$.}
\STATE{Predictor: 
\begin{equation*}
\begin{split}
\frac{w_{i,j}^{n+1}-w_{i,j}^n}{\Delta t}+[\frac{1}{2}(q_{i,j}^{n+1})^2 L_{xx}+(q_{i,j}^{n+1})\rho \sqrt{\delta}L_{xz}
+\delta(\frac{1}{2}L_{zz} +\kappa\theta L_{z1}- \kappa L_{z2})] (\frac{w_{i,j}^{n+1}+w_{i,j}^n}{2})=0
\end{split}
\end{equation*}
with 
\begin{equation*}
\begin{split}
q_{i,j}^{n+1}=&u\mathbbm{1}_{\{L_A(w_{i,j}^{n+1})=\max\{L_A, L_B,L_C\}(w_{i,j}^{n+1})\}}+ d\mathbbm{1}_{\{L_B(w_{i,j}^{n+1})=\max\{L_A, L_B,L_C\}(w_{i,j}^{n+1})\}}\\
&-\frac{\rho \sqrt{\delta} L_{xz}}{L_{xx}}(w_{i,j}^{n+1})\mathbbm{1}_{\{L_C(w_{i,j}^{n+1})=\max\{L_A, L_B,L_C\}(w_{i,j}^{n+1})\}}
\end{split}
\end{equation*}}

\STATE{Corrector: 
\begin{equation*}
\begin{split}
\frac{w_{i,j}^{n+1}-w_{i,j}^n}{\Delta t}+[\frac{1}{2}(q_{i,j}^{n})^2 L_{xx}+(q_{i,j}^{n})\rho \sqrt{\delta}L_{xz}
+\delta(\frac{1}{2}L_{zz} +\kappa\theta L_{z1}
- \kappa L_{z2})] (\frac{w_{i,j}^{n+1}+w_{i,j}^n}{2})=0
\end{split}
\end{equation*}
with 
\begin{equation*}
\begin{split}
q_{i,j}^{n}=&u\mathbbm{1}_{\{L_A(\frac{w_{i,j}^{n+1}+w_{i,j}^n}{2})=\max\{L_A, L_B,L_C\}(\frac{w_{i,j}^{n+1}+w_{i,j}^n}{2})\}}+ d\mathbbm{1}_{\{L_B(\frac{w_{i,j}^{n+1}+w_{i,j}^n}{2})=\max\{L_A, L_B,L_C\}(\frac{w_{i,j}^{n+1}+w_{i,j}^n}{2})\}}\\
&-\frac{\rho \sqrt{\delta} L_{xz}}{L_{xx}}(\frac{w_{i,j}^{n+1}+w_{i,j}^n}{2})\mathbbm{1}_{\{L_C(\frac{w_{i,j}^{n+1}+w_{i,j}^n}{2})=\max\{L_A, L_B,L_C\}(\frac{w_{i,j}^{n+1}+w_{i,j}^n}{2})\}}
\end{split}
\end{equation*}}

\end{algorithmic}
\end{breakablealgorithm}
\bigskip
We set $\kappa=15$ and $\theta=0.04$, which satisfies the Feller condition required in this paper.

\subsection{Error analysis}
To visualize the approximation as $\delta$ vanishes, we plot $P^{\delta}$, $P_0$ and $P_0+\sqrt{\delta}P_1$ with ten equally
spaced values of $\delta$ from $0.05$ to $0$, and consider a typical case of correlation $\rho=-0.9$ (see \cite{in2010adi}). In Figure \ref{fig:four_plots}, we see that the first order prices capture the main feature of the worst case scenario prices for different values of $\delta$. As can be seen, for $\delta$ very small, the approximation performs very well and it worth noting that, even for $\delta$ not very small such as $0.1$, it still performs well.

\begin{figure}
  \centering
  \label{fig:four_plots}
  \includegraphics[height=6in,width=6in]{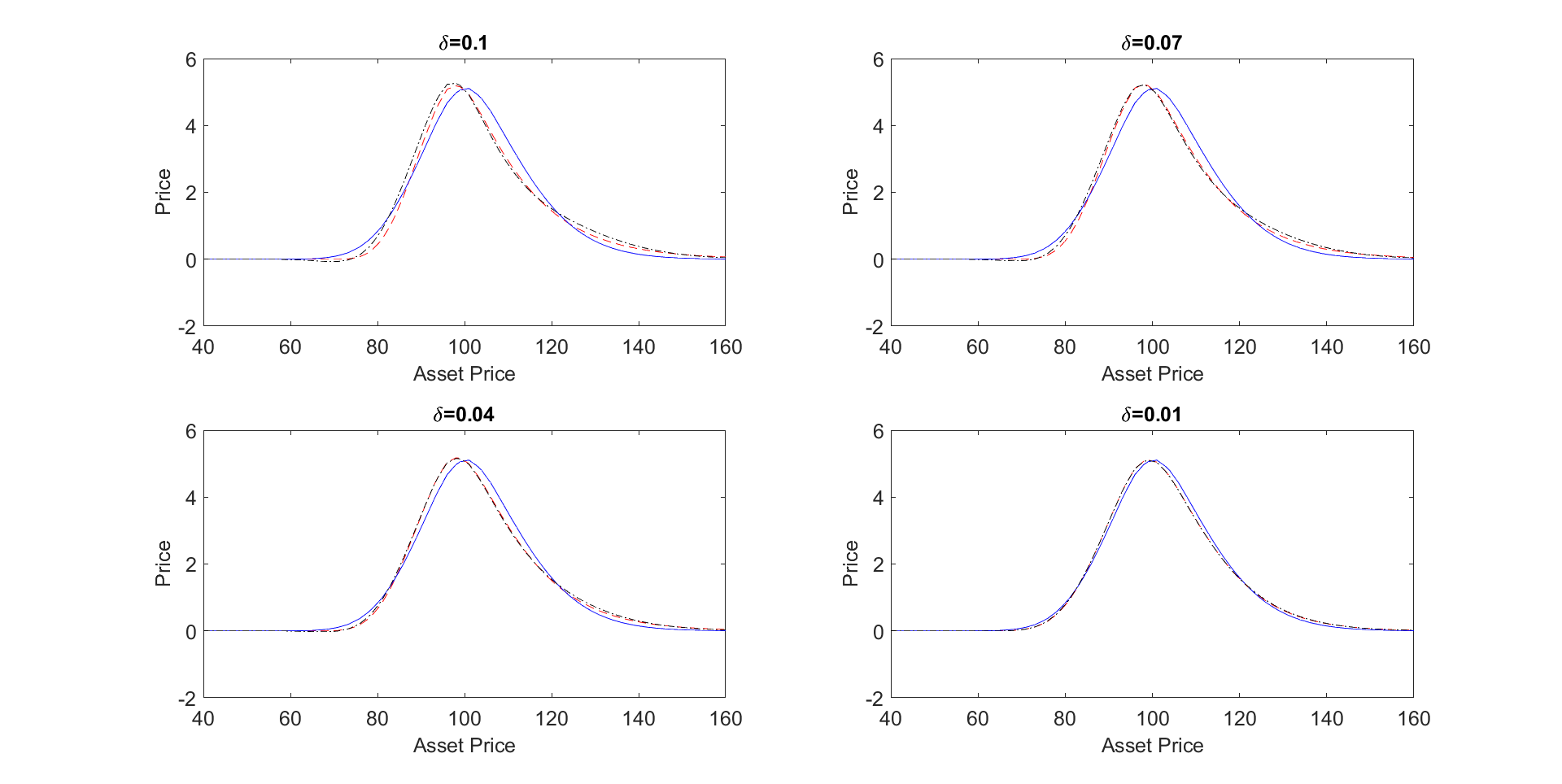}.   
  \caption{\it The red curve marked with "- -" represents the worst case scenario prices $P^{\delta}$; the blue curve represents the leading term $P_0$; the black curve marked with '-.' represents the approximation $P_0+\sqrt{\delta}P_1$.}
\end{figure}

\begin{figure}[htbp]
  \centering
  \label{fig:error}
  \includegraphics[height=3in,width=5in]{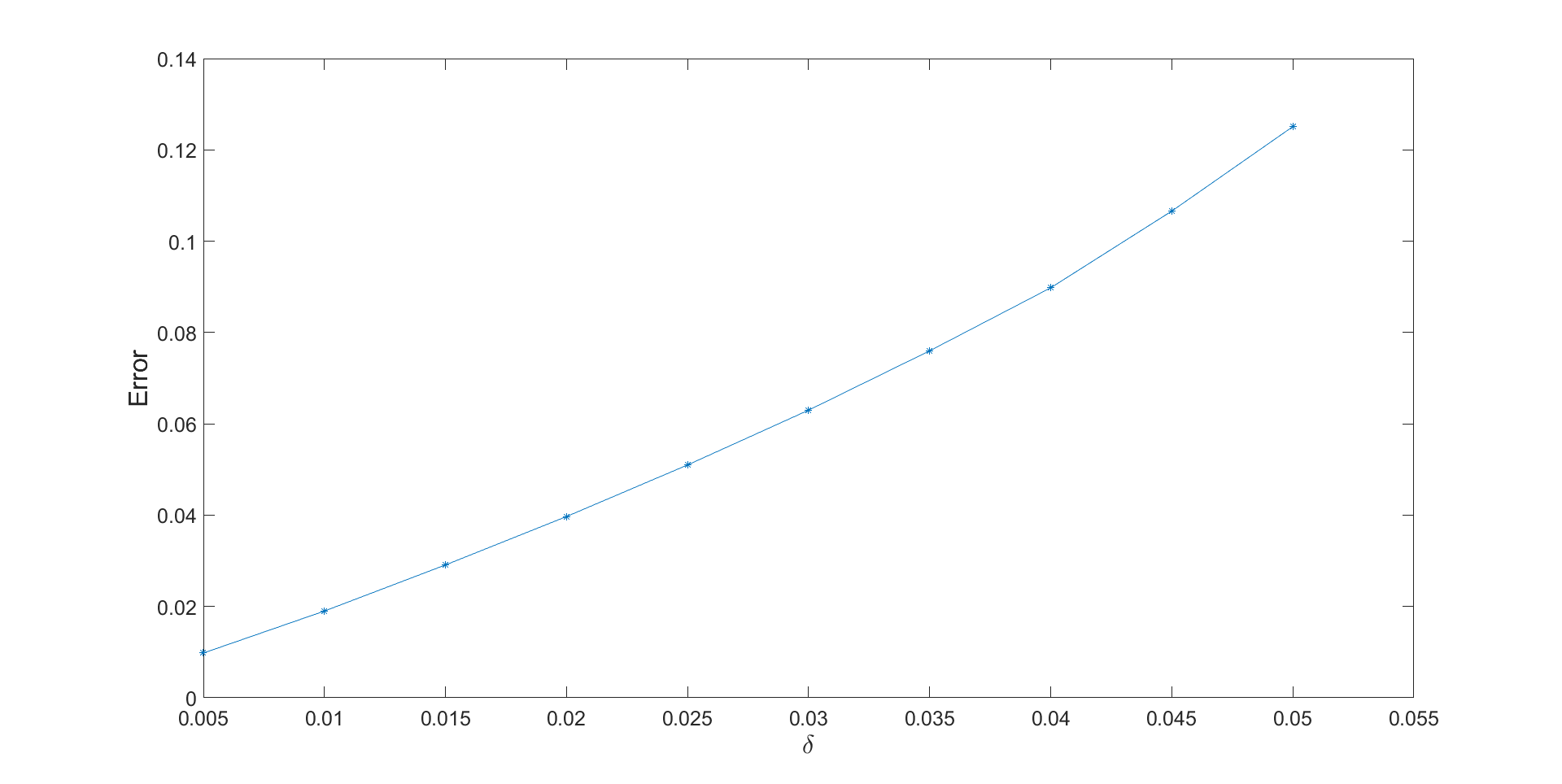}.   
  \caption{\it Error for different values of $\delta$}
\end{figure}

\begin{figure}[htbp]
  \centering
  \label{fig:second_derivative}
  \includegraphics[height=3in,width=5in]{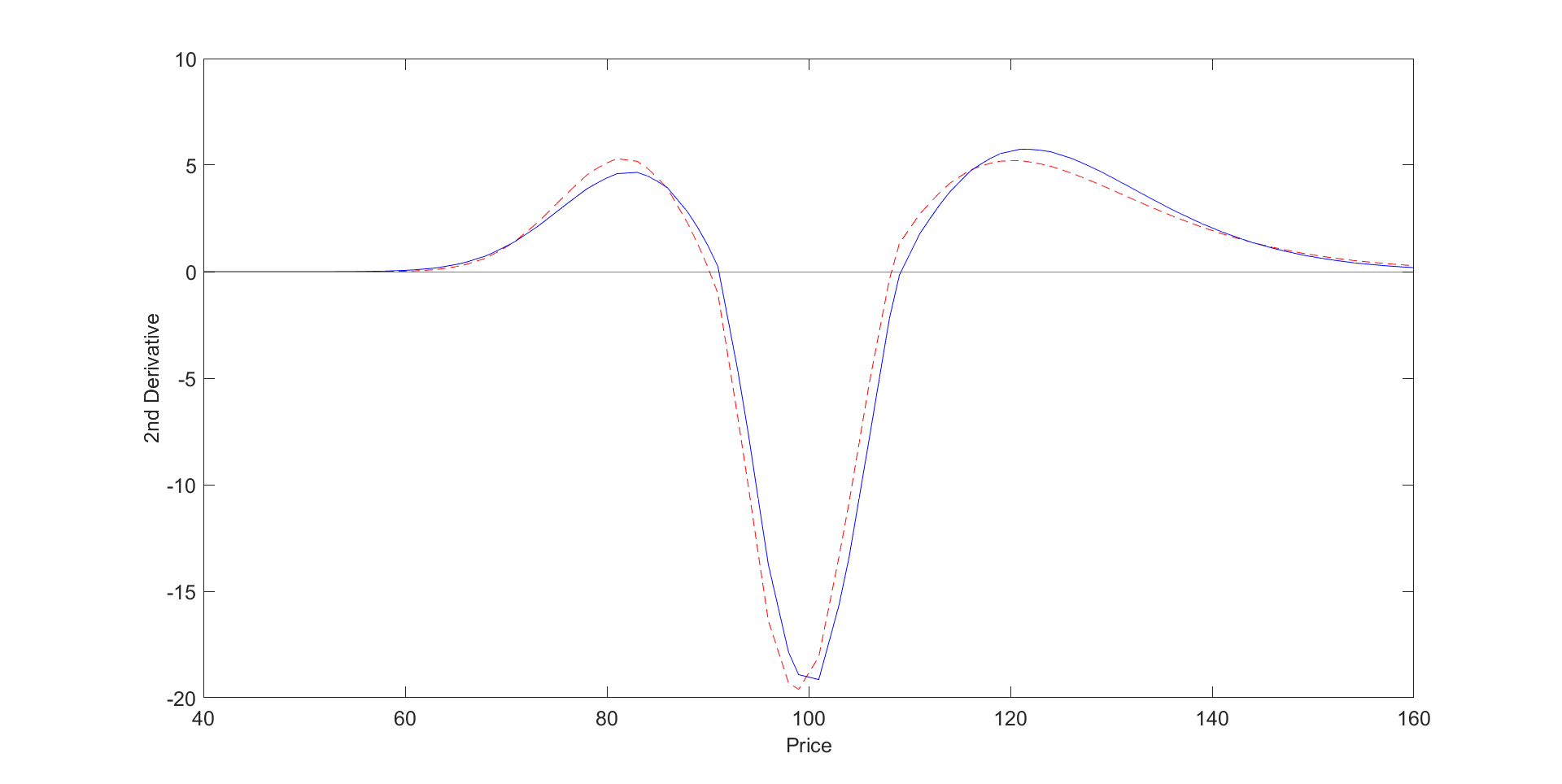}.   
  \caption{\it The red curve marked with "- -" represents $\partial_{xx}^2P^{\delta}$; the blue curve represents $\partial_{xx}^2P_0$.}
\end{figure}

To investigate the convergence of the error of our approximation as $\delta$ decrease, we compute the error of the approximation for each $\delta$ as following
$$\text{error}(\delta)=\sup_{x,z} |P^{\delta}(0,x,z)-P_0(0,x,z)-\sqrt{\delta}P_1(0,x,z)|.$$
As shown in Figure \ref{fig:error}, the error decreases linearly as $\delta$ decreases (at least for $\delta$ small enough), as predicted by  our Main Theorem \ref{main_theorem}.

\begin{remark}
In Remark \ref{remark:thrid_derivative}, for the case that $P_0$ has a third derivative with respect to $x$, which does not vanish on the set $S_{t,z}^{0}$, we have Assumption \ref{zero_points} (ii) as a direct result. In Figure \ref{fig:second_derivative}, we can see that the slopes at the zero points of $\partial_{xx}^2P^{\delta}$ and $\partial_{xx}^2P_0$ are not $0$, hence for this symmetric butterfly spread, Assumption \ref{zero_points} (ii) is satisfied.
\end{remark}

\section{Conclusion}
\label{sec:conclusion}
In this paper, we have proposed the uncertain volatility models with stochastic bounds driven by a CIR process. Our method  is not limited to the CIR process and can be used with any other positive stochastic processes such as positive functions of an OU process. We further studied the asymptotic behavior of the worst case
scenario option prices in the regime of slowly varying stochastic bounds. This study not only helps
 understanding that  uncertain volatility models with stochastic bounds are more flexible than uncertain volatility models with constant bounds for option pricing and  risk management, but also provides 
an approximation procedure for worst-case scenario option prices when the bounds are slowly varying. From the
numerical results, we see that the approximation procedure works really well even when the payoff function does not satisfy the requirements enforced in this paper, and even when $\delta$ is not so small such as $\delta=0.1$. 

Note that as risk evaluation in a financial management requires more accuracy and efficiency nowadays, our approximation procedure highly improves the estimation and still maintains the 
same efficiency level as the regular uncertain volatility models.
Moreover, the worst case scenario price $P^{\delta}$ \eqref{BSB} has to be recomputed for any change in its parameters $\kappa$, $\theta$ and $\delta$. However, the PDEs \eqref{P0} and \eqref{P1} for $P_0$ and $P_1$
are independent of these parameters, so the approximation requires  only to compute $P_0$ and $P_1$ once for
all values of $\kappa$, $\theta$ and $\delta$.

%%%%%%%%%%%%%%%%%%%%%%%%%%%%%%%%%%%%%%%%%%%%%%%%%%%%%%%%%%%%%%%%%%%%%%%%%%%%

\appendix
\section{Moments of $Z_t$ and $X_t$}
\label{moment_result}
\begin{proposition} 
The process $Z$ has finite moments of any order uniformly in $0\leq \delta \leq 1$ for $t\leq T$.
\end{proposition}

The proof is given by Lemma 4.9 in \cite{fouque2011multiscale}. Thus, 
\begin{equation}
\label{Z_moment}
\mathbb{E}_{(t,x,z)}\left[\int_t^T|Z_s|^k ds\right]\leq \mathbb{E}_{(0,z)}\left[\int_0^T|Z_s|^k ds\right] \leq C_k(T,z),
\end{equation}
where $C_k(T,z)$ may depend on ($k, T, z$) but not on $\delta$.

Denote the moment generating function of the integrated CIR process given $Z_s|_{s=t}=z$ as 
$$M_z^{\delta}(\eta):=\mathbb{E}_{(t,z)}[\exp(\eta \int_0^t Z_s ds)],\;\;\;\; \text{for }\eta \in \mathbb{R},$$
and then we have the following lemma:
\begin{proposition}
\label{integrated_moment}
For $\eta \in \mathbb{R}$ independent of $\delta$ and $t$, $M_z^{\delta}(\eta)$ is bounded uniformly in $0\leq \delta \leq 1$ and $t\leq T$. That is $|M_z^{\delta}(\eta)|\leq N(T,z,\eta) <\infty$, where $N(T,z,\eta)$ is independent of $\delta$ and $t$.
\end{proposition}
\begin{proof}
The moment generating function of the integrated CIR process has an explicit form, which is presented in Section 5 of  \cite{lepage2006continuous}. That is
\begin{equation*}
\begin{split}
M_z^{\delta} (\eta)=\Psi(\eta, t)e^{-z\Xi(\eta, t)},
\end{split}
\end{equation*}
where
\begin{equation*}
\begin{split}
\Psi(\eta, t)&=\left( \frac{\bar{b} e^{\delta \kappa \frac{t}{2}} }{\bar{b} \frac{e^{\bar{b} \frac{t}{2}}+e^{-\bar{b} \frac{t}{2}}}{2}+\delta \kappa \frac{e^{\bar{b} \frac{t}{2}}-e^{-\bar{b} \frac{t}{2}}}{2}} \right )^{2\kappa},
\end{split}
\end{equation*}
\begin{equation*}
\begin{split}
\Xi(\eta, t)&=\left( \frac{2\eta \frac{e^{\bar{b} \frac{t}{2}}-e^{-\bar{b} \frac{t}{2}}}{2}  }{\bar{b} \frac{e^{\bar{b} \frac{t}{2}}+e^{-\bar{b} \frac{t}{2}}}{2}+\delta \kappa \frac{e^{\bar{b} \frac{t}{2}}-e^{-\bar{b} \frac{t}{2}}}{2}} \right )^{2\kappa},
\end{split}
\end{equation*}
 and $$\bar{b}=\sqrt{b^2-2\eta \sigma^2}=\sqrt{\delta^2\kappa^2-2\eta\delta}.$$
 In the following, we are going to show $|M_z^{\delta}(\eta)|\leq N(T,z,\eta) <\infty$, where $N(T,z,\eta)$ is independent of $\delta$ and $t$.

\begin{itemize}
\item If $\delta^2\kappa^2-2\eta\delta\geq 0$, we have $\bar{b}\geq 0$ and
\begin{align*}
\Psi(\eta, t)&\leq \left( \frac{\bar{b} e^{\delta \kappa \frac{t}{2}} }{\bar{b} \frac{e^{\bar{b} \frac{t}{2}}+e^{-\bar{b} \frac{t}{2}}}{2} }  \right)^{2\kappa} & 
( \delta \kappa \frac{e^{\bar{b} \frac{t}{2}}-e^{-\bar{b} \frac{t}{2}}}{2}\geq 0)\\
& \leq (e^{\delta \kappa \frac{t}{2}})^{2\kappa}
&( \frac{e^{\bar{b} \frac{t}{2}}+e^{-\bar{b} \frac{t}{2}}}{2}\geq 1)\\
& \leq (e^{\kappa \frac{T}{2}})^{2\kappa}.
\end{align*}
Since $\Xi(\eta, t) \geq 0$, we have $e^{-z\Xi(\eta, t)}\leq 1$. Therefore
$$M_z^{\delta}(\eta)=\Psi(\eta, t)e^{-z\Xi(\eta, t)}\leq (e^{\kappa \frac{T}{2}})^{2\kappa}. $$
\item  If $\delta^2\kappa^2-2\eta\delta< 0$, let $v=\sqrt{2\eta\delta-\delta^2\kappa^2}$ which is positive, then
\begin{equation*}
\begin{split}
M_z^{\delta} (\eta)&=\Psi(\eta, t)e^{-z\Xi(\eta, t)}\\
&=\left(\frac{i v e^{\delta \kappa \frac{t}{2} }}{i v \cos(\frac{vt}{2})+\delta \kappa i \sin( \frac{vt}{2}) } \right)^{2\kappa} \exp \left[-z (\frac{2\eta i \sin( \frac{vt}{2}) }{i v \cos(\frac{vt}{2})+\delta \kappa i \sin( \frac{vt}{2}) } )^{2\kappa} \right]\\
&=\left(\frac{v e^{\delta \kappa \frac{t}{2} }}{v \cos(\frac{vt}{2})+\delta \kappa \sin( \frac{vt}{2}) } \right)^{2\kappa} \exp \left[-z (\frac{2\eta \sin( \frac{vt}{2}) }{v \cos(\frac{vt}{2})+\delta \kappa \sin( \frac{vt}{2}) } )^{2\kappa} \right].
\end{split}
\end{equation*}
\begin{itemize}
 \item In the limit of small $v$, since $$\left(\frac{2\eta \sin( \frac{vt}{2}) }{v \cos(\frac{vt}{2})+\delta \kappa \sin( \frac{vt}{2}) } \right)^{2\kappa}\geq 0), $$ we have\\
\begin{equation*}
\begin{split}
M_z^{\delta}(\eta)&\leq \left(\frac{v e^{\delta \kappa \frac{t}{2} }}{v \cos(\frac{vt}{2})+\delta \kappa \sin( \frac{vt}{2}) } \right)^{2\kappa}\\
&=\left(\frac{v e^{\delta \kappa \frac{t}{2} }}{v (1+\mathcal{O}(v^2t^2 ) )+\delta \kappa ( \frac{vt}{2} + \mathcal{O}(v^3t^3) ) } \right)^{2\kappa}\\
&=\left(\frac{ e^{\delta \kappa \frac{t}{2} }}{1 + \frac{\delta \kappa t}{2} +\mathcal{O}(v^2t^2 ) } \right)^{2\kappa}.
\end{split}
\end{equation*}
There exists $v_0$ independent of $\delta$ and $t$, such that for $v <v_0$,  
$M_z^{\delta}(\eta)\leq \left(\frac{ e^{\kappa \frac{T}{2} }}{\frac{1}{2} } \right)^{2\kappa}.$

Note that $v <v_0$ corresponds to a open region of $\delta$, namely $\delta \in I_1 \subset [0,1]$. 

 \item Consider $\delta \in [0,1]\setminus I_1$, and then consider $v \cos(\frac{vt}{2})+\delta \kappa \sin( \frac{vt}{2})$ in a small neighborhood of $0$. Near these points, notice that $v$ is not in a small neighborhood of $0$, which implies that $\sin(\frac{vt}{2})$ is not in a small neighborhood of $0$.
 
From $\frac{1}{x^n} e^{-\frac{a}{x}}\rightarrow 0, \,\text{ as }\,x \rightarrow 0, \,\text{ for any }\, a>0$ applied to 
$x=v \cos(\frac{vt}{2})+\delta \kappa \sin( \frac{vt}{2})$, we deduce that  there exists an open subset $I_2 \subset [0,1] \setminus I_1$, such that $M_z^{\delta}(\eta) \leq M_1(T,z,\eta)$, which is independent of $\delta$ and $t$.

 \item Lastly, on $[0,1]\setminus (I_1 \cup I_2)$, which is a closed region and thus compact, $M_z(\eta)$ is well-defined and continuous with respect to $\delta$ and $t$. So there exists $M_2(T,z,\eta)$ independent of $\delta$ and $t$, such that $|M_z^{\delta}(\eta)|\leq M_2(T,z,\eta)$.
\end{itemize}
\end{itemize}

In sum, taking $N=\max\{(2e^{\kappa \frac{T}{2}})^{2\kappa}, M_1, M_2\}$, which is independent of $\delta$ and $t$, we have 
$|M_z^{\delta}(\eta)|\leq N(T,z,\eta)\, \text{(the uniform bound)}$,
as desired.
\end{proof}

By this result, we have the following Proposition:
\begin{proposition}
The process $X$ has finite moments of any order uniformly in $0\leq \delta \leq 1$ for $t\leq T$.
\end{proposition}
\begin{proof}
For the process $X_t$ satisfying the  SDE
\begin{equation*}
dX_t=rX_tdt+q_t\sqrt{Z_t} X_tdW_t,
\end{equation*}
where $q_t \in [d, u]$ and $Z_t$ is the CIR process satisfies SDE \eqref{processZ},
for each finite $n\in \mathbb{Z}$, we have
\begin{equation*}
\begin{split}
X_{s+t}^n =& x^n \exp \left(nrs-\frac{n}{2}\int_t^{s+t} (q_v\sqrt{Z_v})^2 dv+n\int_t^{s+t}{q_v\sqrt{Z_v}}dW_v \right)\\
=&x^n \exp \left(nrs+\frac{n^2-n}{2}\int_t^{s+t} (q_v\sqrt{Z_v})^2 dv \right)\\
&\times \exp \left(-\frac{n^2}{2}\int_t^{s+t} (q_v\sqrt{Z_v})^2 dv+n\int_t^{s+t} {q_v\sqrt{Z_v}}dW_v \right)\\
\leq &x^n \exp \left(nrs+\frac{n^2-n}{2}\int_t^{s+t} u^2 Z_v dv\right)\\
&\times \exp \left(-\frac{n^2}{2}\int_t^{s+t} (q_v\sqrt{Z_v})^2 dv+n\int_t^{s+t} {q_v\sqrt{Z_v}}dW_v \right),
\end{split}
\end{equation*}
and the last step follows by the inequality $q_v < u$.

The Novikov condition is satisfied thanks to Proposition \ref{integrated_moment}, that is,
\begin{equation*}
\begin{split}
\mathbb{E}_{(t,x,z)}\left[\exp \left(\frac{1}{2}\int_t^{s+t} (nq\sqrt{Z_v})^2 dv \right)\right] &\leq \mathbb{E}_{(t,x,z)}\left[\exp \left(\frac{n^2u^2}{2}\int_t^{s+t} Z_v dv \right)\right]\\
&=M_z^{\delta}( \frac{n^2u^2}{2})<\infty.
\end{split}
\end{equation*}

Now we know that
$$\Lambda_s=\exp\left(-\frac{n^2}{2}\int_t^{s+t} (q_v\sqrt{Z_v})^2 dv+n\int_t^{s+t} {q_v\sqrt{Z_v}}dW_v \right)$$
is a martingale.

Hence,
\begin{align*}
\mathbb{E}_{(t,x,z)}[X_{s+t}^n]&\leq x^n \exp(nrs)\mathbb{E}_{(t,x,z)}\left[\exp(\frac{(n^2-n)u^2 }{2}\int_t^{s+t} Z_v dv)\right]\\
&= x^n \exp(nrs)M_z^{\delta} \left(\frac{(n^2-n)u^2 }{2} \right) \\
& \leq x^n \exp(nrs) N(T,z, \frac{(n^2-n)u^2 }{2}), \\
& \leq x^n \exp(nrs) N(T,z, \frac{(n^2-n)u^2 }{2}), &(\text{Proposition} \ref{integrated_moment})
\end{align*}
where the upper bound $x^n \exp(nrs) N(T,z, \frac{(n^2-n)u^2 }{2})$ is independent of $\delta$ and $t$. 
\end{proof}
Therefore,
\begin{equation}
\mathbb{E}_{(t,x,z)}[\int_t^T|X_s|^k ds]\leq \mathbb{E}_{(0,x,z)}[\int_0^T|X_s|^k ds] \leq N_k(T,x,z),
\end{equation}
where $N_k(T,x,z)$ may depend on ($k, T, x, z$) but not on $0\leq \delta\leq 1$.

\section{Proof of Proposition \ref{X_cvg}}
\label{appendix_X_cvg}
Integrating over $[t,T]$ the SDE \eqref{processXdelta} and the SDE \eqref{processX0}, we have 	
\begin{equation}
\label{Int_processXdelta}
X_{T}^{\delta}=x+\int_t^{T} rX_s^{\delta}ds+\int_t^{T} q_s \sqrt{Z_s} X_s^{\delta}dW_s
\end{equation}
and 
\begin{equation}
\label{Int_processX0}
X_{T}^{0}=x+\int_t^{T} rX_s^{0}ds+\int_t^{T} q_s \sqrt{z} X_s^{0}dW_s.
\end{equation}
%where $X_t^{\delta}=X_t^{0}=x$.

The difference of \eqref{Int_processXdelta} and \eqref{Int_processX0} is given by
\begin{equation}
\begin{split}
X_{T}^{\delta}-X_{T}^{0}=&\int_t^{T} r(X_s^{\delta}-X_s^{0})ds+\int_t^{T} q_s (\sqrt{Z_s} X_s^{\delta}-\sqrt{z} X_s^{0})dW_s\\
=&\int_t^{T} r(X_s^{\delta}-X_s^{0})ds+\int_t^{T} q_s \sqrt{z} (X_s^{\delta}-X_s^{0})dW_s+\int_t^{T} q_s (\sqrt{Z_s} -\sqrt{z})X_s^{\delta}dW_s.
\end{split}
\end{equation}
Let $Y_s=X_s^{\delta}-X_s^{0}$, then $Y_t=0$ and
\begin{equation}
Y_{T}=\int_t^{T} rY_sds+\int_t^{T} q_s \sqrt{z}Y_sdW_s+\int_t^{T} q_s (\sqrt{Z_s} -\sqrt{z})X_s^{\delta}dW_s.
\end{equation}
Therefore,
\begin{equation}
\label{expectation_Y_T1_square}
\begin{split}
&\mathbb{E}_{(t,x,z)}[Y_{T}^2]\\
\leq &3\mathbb{E}_{(t,x,z)}[(\int_t^{T} rY_sds)^2+(\int_t^{T} q_s \sqrt{z}Y_sdW_s)^2+(\int_t^{T} q_s (\sqrt{Z_s} -\sqrt{z})X_s^{\delta}dW_s)^2]\\
\leq &  \int_t^{T} (3T r^2+3u^2z)\mathbb{E}_{(t,x,z)}[Y_s^2] ds+\underbrace{3u^2 \int_t^{T} \mathbb{E}_{(t,x,z)}[(\sqrt{Z_s} -\sqrt{z})^2 (X_s^{\delta})^2]ds}_{R(\delta)}.
\end{split}
\end{equation}
Notice that only the upper bound of $q$ is used, which gives the uniform convergence in $q$. Also note that using the result that $X_t$ and $Z_t$ have finite moments uniformly in $\delta$, we can show that $|R(\delta)|\leq C\delta$,
where $C=C(T,\theta,u,d,z)$ is independent of $\delta$.

Denote $f({T})=\mathbb{E}_{(t,x,z)}(Y_{T}^2)$ and $\lambda=3T r^2+3u^2z >0$,
then equation \eqref{expectation_Y_T1_square} can be written as 
%& \text{ ()}
\begin{equation*}
f({T})\leq  \int_t^{T} \lambda f(s)ds+C \delta \leq  \delta \int_t^{T} C\lambda  e^{\lambda(T-s)}ds+C \delta
\end{equation*}
Therefore, by Gronwall inequality,
\begin{equation*}
\mathbb{E}_{(t,x,z)} (X_{T}^{\delta}-X_{T}^{0})^2=\mathbb{E}_{(t,x,z)}Y_{T}^2=f(T)\leq C'\delta,
\end{equation*}
and the Proposition follows.

\section{Existence and uniqueness of ($X_t^{\ast, \delta}$)}
\label{Existence_and_uniqueness}
For the existence and uniqueness of $X_t^{\ast,\delta}$, we consider the transformation
\begin{equation*}
Y_t^{\ast,\delta}:=\log X_t^{\ast,\delta},
\end{equation*} 
which is well defined for any $t<\tau^{\epsilon}$, where for any $\epsilon >0$,
\begin{equation*}
\begin{split}
\tau^{\epsilon}:&=\inf\{t>0 | X_t^{\ast,\delta}=\epsilon \;\;\text{ or }\;\; X_t^{\ast,\delta}=1/\epsilon  \}\\
&= \inf\{t>0| Y_t^{\ast,\delta}=\log \epsilon \;\;\text{ or }\;\; Y_t^{\ast,\delta}=-\log \epsilon  \}.
\end{split}
\end{equation*} 
By It\^{o}'s formula, the process $Y_t^{\ast,\delta}$ satisfies the following SDE:
\begin{equation}
\label{processY_ast}
dY_t^{\ast,\delta}=-\frac{1}{2}(q^{\ast,\delta})^2 Z_t dt + q^{\ast,\delta}\sqrt{Z_t}dW_t.
\end{equation} 
Note that the diffusion coefficient satisfies
$$q^{\ast,\delta}\sqrt{Z_t} \geq d\sqrt{Z_t} >0,$$ 
and is bounded away from $0$, hence by Theorem 1 in section 2.6 of \cite{krylov2008controlled} and the result 7.3.3 of \cite{stroock2007multidimensional}, the SDE \eqref{processY_ast} has a unique weak solution. Consequently, we have a unique solution to the SDE \eqref{processX_ast} until $\tau^{\epsilon}$ for any $\epsilon>0$. In order to show \eqref{processX_ast} has a unique solution, it suffices to prove that 
\begin{equation*}
\lim_{\epsilon \downarrow 0}\mathbb{Q}(\tau^{\epsilon}<T)=0,
\end{equation*} 
for any $T>0$.

First, for any $t \in [0, T]$,
\begin{equation*}
Y_t^{\ast,\delta}=\int_0^t -\frac{1}{2}(q^{\ast,\delta})^2 Z_s ds + \int_0^t q^{\ast,\delta}\sqrt{Z_s}dW_s.
\end{equation*} 
Then
\begin{equation*}
\begin{split}
\mathbb{E} |Y_t^{\ast,\delta}| &\leq \mathbb{E} \left [\int_0^t \frac{1}{2}(q^{\ast,\delta})^2 Z_s ds \right ] + \mathbb{E} \left [\int_0^t q^{\ast,\delta}\sqrt{Z_s}dW_s \right ]\\
&=A+B.
\end{split}
\end{equation*} 
For term A, by equation \eqref{Z_moment}, we have 
\begin{equation*}
A \leq \frac{1}{2} u^2 \int_0^t \mathbb{E} [Z_s] ds \leq \frac{1}{2} u^2 C_1(T,z).
\end{equation*}
For term B, we have,
\begin{align*}
B &\leq \left(\mathbb{E} \left ( \int_0^t q^{\ast,\delta}\sqrt{Z_s}  dW_s\right )^2 \right )^{1/2}  & (\text{Cauchy--Schwartz inequality})\\
&=\left( \int_0^t \mathbb{E} \left((q^{\ast,\delta})^2 Z_s \right ) ds \right )^{1/2}  &(\text{It\^{o} isometry})\\
&\leq \left( \int_0^t u^2\mathbb{E} (Z_s) ds \right )^{1/2} & (d\leq q^{\ast,\delta} \leq u )\\
& \leq u \sqrt{C_1(T,z)} & (\text{equation}  \eqref{Z_moment}).
\end{align*}
Hence, there exists a constant $D=D(T,z,u)$ such that $\mathbb{E} |Y_t^{\ast,\delta}| \leq D$.

Furthermore, 
\begin{equation*}
\begin{split}
\mathbb{Q}( \sup_{t \in [0, T]} |Y_t^{\ast,\delta}|>|\log \epsilon|  ) &= 1-\mathbb{Q}( |Y_t^{\ast,\delta}| \leq |\log \epsilon|\;\text { for all }0\leq t \leq T)\\
& \leq 1-\mathbb{Q}\left( |Y_t^{\ast,\delta}- \mathbb{E} [Y_t^{\ast,\delta}] | \leq \log |\epsilon|-|\mathbb{E} [Y_t^{\ast,\delta}] |\;\text { for all }0\leq t \leq T\right)\\
\end{split}
\end{equation*}
Here, take $\epsilon$ small enough, depending on $T$, $z$ and $u$, such that
\begin{align*}
\mathbb{Q}( \sup_{t \in [0, T]} |Y_t^{\ast,\delta}|>|\log \epsilon|  ) 
& \leq 1-\mathbb{Q}\left( |Y_t^{\ast,\delta}- \mathbb{E} Y_t^{\ast,\delta}| \leq \frac{\log |\epsilon|}{2}\; \text { for all }0\leq t \leq T \right)\\
& = \mathbb{Q}\left(\sup_{t \in [0, T]}|Y_t^{\ast,\delta}- \mathbb{E} Y_t^{\ast,\delta}| > \frac{\log |\epsilon|}{2}\right )\\
& \leq \frac{ \Var( Y_t^{\ast,\delta}) }{(\frac{\log |\epsilon|}{2})^2},
\end{align*}
where the last step follows from Doob--Kolmogorov inequality.

Last, 
\begin{align*}
\Var( Y_t^{\ast,\delta})&=\mathbb{E} \left( \int_0^t q^{\ast,\delta}\sqrt{Z_s}dW_s  \right)^2 \\
& =\int_0^t \mathbb{E} \left( (q^{\ast,\delta})^2 Z_s\right)ds  & (\text{It\^{o} isometry})\\\
& \leq u^2 C_1(T,z)  & (d\leq q^{\ast,\delta} \leq u \text{ and equation  \eqref{Z_moment}}).  
\end{align*}

Finally, as $\epsilon \downarrow 0$, we can conclude that 
\begin{equation*}
\lim_{\epsilon \downarrow 0}\mathbb{Q}(\tau^{\epsilon}<T)=0,
\end{equation*} 
for any $T>0$, as desired.
%\begin{lemma}
%Let ($X_t^{\ast, \delta}$) satisfies the SDE  \eqref{processX_ast}, then for any $t \in (0,T]$, $X_t^{\ast, \delta}$ does not weight points.
%\end{lemma}
%\begin{proof}
%If we want to show $X_t^{\ast, \delta}$ does not weight points, it suffices to show that $Y_t^{\ast, \delta}$ does not weight points.
%\end{proof}

\section{Proof of \eqref{term_sum}}
\label{Appendix:proof_of_eq_term_sum}
From \eqref{eq1} and \eqref{tau_inverse_upper_bound}, by decomposing in $\{\sup_{t\leq s \leq T}Z_s \leq M \}$ and $\{\sup_{t\leq s \leq T}Z_s > M \}$ for any $M>z$, we have
\begin{equation}
\begin{split}
\mathbb{E}_{(t,x,z)} \bigg [\int_t^{\tau^{-1}(T)} \mathbbm{1}_{\{|B_v-x_i|<C\sqrt{\delta}\}}dv \bigg ] &\leq \mathbb{E}_{(t,x,z)} \bigg [\int_t^{Du^2T \sup_{t\leq s \leq T}Z_s} \mathbbm{1}_{\{|B_v-x_i|<C\sqrt{\delta}\}}dv \bigg ]\\
&\doteq \circled{1}+\circled{2}.
\end{split}
\end{equation}
Then
\begin{equation}
\label{eq:term1_result}
\begin{split}
\circled{1} \leq & \mathbb{E}_{(t,x,z)} \bigg [\int_t^{Du^2TM} \mathbbm{1}_{\{|B_v-x_i|<C\sqrt{\delta}\}} \mathbbm{1}_{\{\sup_{t\leq s \leq T}Z_s \leq M \}} dv \bigg ]\\
\leq & \mathbb{E}_{(t,x,z)} \bigg [\int_t^{Du^2TM} \mathbbm{1}_{\{|B_v-x_i|<C\sqrt{\delta}\}} dv \bigg ]\\
=& \int_t^{Du^2TM} \mathbb{Q}^B{\{|B_v-x_i|<C\sqrt{\delta}\}} dv\\
\leq & \int_t^{Du^2TM} \frac{2C\sqrt{\delta}}{\sqrt{2\pi v}} dv\\
\leq & \sqrt{\delta} (\frac{4C}{\sqrt{2\pi}} \sqrt{Du^2TM}).
\end{split}
\end{equation}
Let $
\Gamma_s=\int_t^s\sqrt{\delta}\sqrt{Z_v}dW_v^Z$, and integrate the SDE of the process $Z$ over $[t,s]$ for $s\in [t,T]$, we obtain
\begin{equation*}
Z_s =z+\int_t^s\delta \kappa(\theta-Z_v)dv+\Gamma_s.
\end{equation*}
Since $Z_t\geq 0$ and $0 \leq \delta \leq 1$, we have
\begin{equation}
\label{eq3}
\begin{split}
\sup_{t\leq s \leq T}Z_s 
&\leq (z+ \kappa\theta T)+\sup_{t\leq s \leq T}\Gamma_s,
\end{split}
\end{equation}
and then let $M=z+\kappa\theta T+1$, we have
\begin{equation}
\label{eq4}
\mathbbm{1}_{\{\sup_{t\leq s \leq T}Z_s>M \}}\leq \mathbbm{1}_{\{\sup_{t\leq s \leq T}\Gamma_s>1 \}}.
\end{equation}
Therefore, by \eqref{eq3} and \eqref{eq4},
\begin{align*}
\circled{2} \leq & \mathbb{E}_{(t,x,z)} \bigg [\int_t^{Cu^2T\sup_{t\leq s \leq T}Z_s}  \mathbbm{1}_{\{\sup_{t\leq s \leq T}Z_s > M \}} dv \bigg ]\\
\leq &\mathbb{E}_{(t,x,z)} \bigg [ \mathbbm{1}_{\{\sup_{t\leq s \leq T}Z_s > M \}} Cu^2T\sup_{t\leq s \leq T}Z_s \bigg ]\\
=& Cu^2T\mathbb{E}_{(t,x,z)} \bigg [\mathbbm{1}_{\{\sup_{t\leq s \leq T}Z_s > M \}}{\sup_{t\leq s \leq T}Z_s}   \bigg ]\\
\leq & Cu^2T(\mathbb{E}_{(t,x,z)}\bigg [(z+ \kappa\theta T)  \mathbbm{1}_{\{\sup_{t\leq s \leq T}\Gamma_s>1 \}} \bigg ] +\mathbb{E}_{(t,x,z)} \bigg [(\sup_{t\leq s \leq T}\Gamma_s)\mathbbm{1}_{\{\sup_{t\leq s \leq T}\Gamma_s>1 \}} \bigg ])\\
\doteq & Cu^2T(\circled{2.1}+\circled{2.2}).
\end{align*}
Note that $\Gamma_s$ is a martingale and thus $\Gamma_s^2$ is a nonnegative submartingale, thus by Doob's martingale inequality and since the process $Z$ has finite moments uniformly in $\delta$,
\begin{equation}
\label{term2.1}
\circled{2.1}=(z+ \kappa\theta T)  \mathbb{Q} \{\sup_{t\leq s \leq T}\Gamma_s^2>1 \}\leq (z+ \kappa\theta T)\mathbb{E}_{(t,x,z)}[\Gamma_T^2] \leq L\delta,
\end{equation}
where $L$ may depend on $z, \kappa, \theta, T$ but not on $\delta$.

Next, by the Doob's martingale inequality in $L^2$ form, we have
\begin{equation}
\label{term2.2}
\circled{2.2}
\leq \mathbb{E}_{(t,x,z)} \bigg [\sup_{t\leq s \leq T}\Gamma_s^2 \bigg ]
\leq 2 \mathbb{E}_{(t,x,z)} [\Gamma_T^2 ]
 \leq L'\delta.
\end{equation}

Finally, by inequalities \eqref{term_sum}, \eqref{eq:term1_result}, \eqref{term2.1} and \eqref{term2.2}, we have
\begin{equation*}
\begin{split}
&\mathbb{E}_{(t,x,z)}\left[\int_t^T \mathbbm{1}_{\{X_s^{\ast, \delta} \in A_{s,z}^{\delta}\}}\sigma^2(X_s^{\ast, \delta}) ds \right] \\
\leq & \sum_{i=1}^n \mathbb{E}_{(t,x,z)} \bigg [\int_t^T \mathbbm{1}_{\{|X_s^{\ast, \delta}-x_i|<C\sqrt{\delta}\}}\sigma^2(X_s^{\ast, \delta}) ds \bigg ]\\
\leq & n[\circled{1}+\circled{2.1}+\circled{2.2}]\\
\leq & n[\sqrt{\delta} (\frac{4C}{\sqrt{2\pi}} \sqrt{Du^2TM})+L\delta+L'\delta]\\
\leq & C_1\sqrt{\delta},
\end{split}
\end{equation*}
where $C_1$ a positive constant and does not depend on $\delta$.

\section{Proof of Uniform Boundedness of $\rom{1}_2$ and $\rom{1}_3$ on $\delta$} 
\label{sec:Proof_of_Uniform_Boundedness_I2_I3}
With the help of Assumption  \ref{derivative_assumptions}, Cauchy--Schwarz inequality and the uniformly bounded moments of $X_t$ and $Z_t$ processes given in Appendix  \ref{moment_result}, we are going to prove that $\rom{1}_2$ and $\rom{1}_3$ are uniformly bounded in $\delta$.

First recall that
\begin{equation*}
\begin{split}
\rom{1}_2=&\mathbb{E}_{(t,x,z)} \Bigg [\int_t^T \Bigg( \rho (q^{\ast, \delta}) Z_s X_s^{\ast, \delta}\partial_{xz}^2P_1(s, X_s^{\ast, \delta},Z_s)\\
&+\frac{1}{2} Z_s\partial_{zz}^2P_0(s, X_s^{\ast, \delta},Z_s) +\kappa(\theta-Z_s)\partial_z P_0(s, X_s^{\ast, \delta},Z_s)\Bigg) ds \Bigg ],
\end{split}
\end{equation*}
and denote 
\begin{equation*}
\rom{1}_2 \doteq \rom{1}_2^{(1)}+\rom{1}_2^{(2)}+\rom{1}_2^{(3)},
\end{equation*}
Then we have
\begin{equation*}
\label{rom1_2}
\begin{split}
\rom{1}_2^{(1)}
\leq & \mathbb{E}_{(t,x,z)} \left[\int_t^T\rho u Z_s X_s^{\ast, \delta}|\partial_{xz}^2P_1(s, X_s^{\ast, \delta},Z_s)|ds \right]\\
\leq & \rho u\mathbb{E}_{(t,x,z)}^{1/2} \left[\int_t^T \left(Z_s X_s^{\ast, \delta}\right)^2ds \right]\cdot \mathbb{E}_{(t,x,z)}^{1/2} \left[\int_t^T \left(\partial_{xz}^2P_1(s, X_s^{\ast, \delta},Z_s) \right )^2ds \right]\\
\leq & \rho u\mathbb{E}_{(t,x,z)}^{1/4} \left[\int_t^T (Z_s)^4 ds \right]\cdot \mathbb{E}_{(t,x,z)}^{1/4} \left[\int_t^T (X_s^{\ast, \delta})^4 ds \right]\\
&\cdot \bar{a}_{11}^2 \mathbb{E}_{(t,x,z)}^{1/2} \left[\int_t^T \left(1+|X_s^{\ast, \delta}|^{\bar{b}_{11}}+|Z_s|^{\bar{c}_{11}}\right )^2ds \right]\\
\leq & \rho u \left(C_4(T,z) \right)^{1/4} \cdot (N_4(T,x,z))^{1/4} 
\cdot \bar{A}_{11} [C_{2\bar{b}_{11}}(T,z)+ N_{2\bar{c}_{11}}(T,x,z)]^{1/2},
\end{split}
\end{equation*}

\begin{equation*}
\label{rom1_3}
\begin{split}
\rom{1}^{(3)}
\leq &\frac{1}{2}\mathbb{E}_{(t,x,z)}^{1/2} \left[\int_t^T (Z_s)^2 ds \right] \cdot \mathbb{E}_{(t,x,z)}^{1/2} \left[\int_t^T \left (\partial_{zz}^2P_0(s, X_s^{\ast, \delta},Z_s) \right )^2 ds \right]\\
\leq & \frac{1}{2} \left(C_2(T,z) \right)^{1/2} \cdot A_{02} [C_{2 b_{02}}(T,z) + N_{2 c_{02}}(T,x,z)]^{1/2}
\end{split}
\end{equation*}
and
\begin{equation*}
\label{rom1_4}
\begin{split}
\rom{1}^{(4)}
&\leq \kappa \mathbb{E}_{(t,x,z)}^{1/2} \left[\int_t^T(\theta-Z_s)^2 ds] \cdot \mathbb{E}_{(t,x,z)}^{1/2}[\int_t^T \left( \partial_z P_0(s, X_s^{\ast, \delta},Z_s) \right )^2ds \right]\\
&\leq \kappa \mathbb{E}_{(t,x,z)}^{1/2} \left[\int_t^T \theta^2+Z_s^2 ds \right] \cdot \mathbb{E}_{(t,x,z)}^{1/2} \left[\int_t^T \left( \partial_z P_0(s, X_s^{\ast, \delta},Z_s) \right )^2ds \right]\\
& \leq  \frac{1}{2} \left(C_2(T,z) +\theta^2(T-t)\right)^{1/2} \cdot A_{01} [C_{2 b_{01}}(T,z) + N_{2 c_{01}}(T,x,z)]^{1/2},
\end{split}
\end{equation*}
where $\bar{A}_{01}$, $\bar{A}_{11}$ and $A_{02}$ are positive constants. 

Next recall that
\begin{equation*}
\rom{1}_3=\mathbb{E}_{(t,x,z)}\Bigg [\int_t^T  \frac{1}{2} Z_s\partial_{zz}^2P_1(s, X_s^{\ast, \delta},Z_s) +\kappa(\theta-Z_s)\partial_zP_1(s, X_s^{\ast, \delta},Z_s) ds \Bigg ],
\end{equation*}
and denote 
\begin{equation*}
\rom{1}_3 \doteq \rom{1}_3^{(1)}+\rom{1}_3^{(2)}.
\end{equation*}
Then we have
\begin{equation*}
\begin{split}
\rom{1}_3^{(1)}
\leq & \frac{1}{2}\mathbb{E}_{(t,x,z)}^{1/2} \left[\int_t^T (Z_s)^2 ds \right] \cdot \mathbb{E}_{(t,x,z)}^{1/2} \left[\int_t^T \left (\partial_{zz}^2P_1(s, X_s^{\ast, \delta},Z_s) \right )^2 ds \right]\\
\leq & (C_2(T,z))^{1/2} \cdot \bar{A}_{02} [C_{2 \bar{b}_{02}}(T,z) + N_{2 \bar{c}_{02}}(T,x,z)]^{1/2}
\end{split}
\end{equation*}
and
\begin{equation*}
\begin{split}
\rom{1}_3^{(2)}
&\leq 2\kappa \mathbb{E}_{(t,x,z)}^{1/2} \left[\int_t^T \theta^2+Z_s^2 ds \right] \cdot \mathbb{E}_{(t,x,z)}^{1/2} \left [\int_t^T \left( \partial_z P_1(s, X_s^{\ast, \delta},Z_s) \right )^2ds \right]\\
&\leq 2\kappa [ \theta^2(T-t)+ C_2(T,z)]^{1/2} \cdot \bar{A}_{01} [C_{2 \bar{b}_{01}}(T,z) + N_{2 \bar{c}_{01}}(T,x,z)]^{1/2},
\end{split}
\end{equation*}
where $\bar{A}_{01}$, $\bar{A}_{02}$ are positive constants.

%\section*{Acknowledgments}
%We would like to acknowledge the assistance of volunteers in putting
%together this example manuscript and supplement.

\bibliographystyle{siamplain}
\bibliography{references}
\end{document}